\documentclass[12pt, a4paper]{amsart}

\usepackage{amsmath,amssymb,amsfonts}
\usepackage{bbold}
\usepackage{epic,eepic}
\usepackage{enumerate}

\usepackage{amsthm}
\usepackage[all]{xy}
\usepackage{caption}
\usepackage{stmaryrd} % for \llbracket
\usepackage[dvipdfmx]{graphicx,color} % no need for Beamer
\usepackage{tikz}
\usetikzlibrary{trees}

\usepackage{geometry}
\usepackage{listings} % for program code font

\theoremstyle{plain}
\newtheorem{thm}{Theorem}[section]
\newtheorem{prop}[thm]{Proposition}
\newtheorem{lem}[thm]{Lemma}

\theoremstyle{definition}
\newtheorem{defn}[thm]{Definition}
\newtheorem{exmp}[thm]{Example}

\newtheorem{rem}[thm]{Remark}
\newtheorem{asm}[thm]{Assumption}

\numberwithin{figure}{section}

\numberwithin{table}{section}

\newcommand{\lspace} {
  \vspace{0.8\baselineskip}
}

\newcommand{\abs}[1]{
  \lvert  #1 \rvert
}

\newdir{ >}{{}*!/-5pt/@{>}}

\definecolor{arrowred}{rgb}{0,0,0} % black
\newcommand{\newword}[1]{\textbf{\textit{#1}}}

\newcommand{\textred}[1]{#1}
\newcommand{\textblue}[1]{#1}

\numberwithin{equation}{section}

\newcommand{\ProbB}{\bar{B}}

\newcommand{\ProbO}{\bar{\Omega}}

\newcommand{\Prob}{\mathbf{Prob}}

\mathchardef\mhyphen="2D  % short hyphen

\title {Generalized Filtrations and Its Application to Binomial Asset Pricing Models}

\thanks{This work was supported by JSPS KAKENHI Grant Number 18K01551.}

\author[T. Adachi, K. Nakajima and Y. Ryu]{Takanori Adachi, Katsushi Nakajima and Yoshihiro Ryu}
\address{Graduate School of Management,
         Tokyo Metropolitan University,
         1-4-1 Marunouchi, Chiyoda-ku, Tokyo 100-0005, Japan}
\email{Takanori Adachi <tadachi@tmu.ac.jp>}

\address{College of International Management,
         Ritsumeikan Asia Pacific University,
         1-1 Jumonjibaru, Beppu, Oita, 874-8577 Japan}
\email{Katsushi Nakajima <knakaji@apu.ac.jp>}

\address{Department of Mathematical Sciences,
         Ritsumeikan University,
         1-1-1 Nojihigashi, Kusatsu, Shiga, 525-8577 Japan}
\email{Yoshihiro Ryu <iti2san@gmail.com>}

\date{\today}

\keywords{
	binomial asset pricing model,
	categorical probability theory,
	generalized filtration
}

\subjclass[2010]{
  Primary 
    91B25,   % Asset pricing models
    16B50;   % Category-theoretic methods and results
  secondary
    60G20,   % Generalized stochastic processes
	91Gxx	% mathematical finacne
}

\begin{document}

\maketitle

\begin{abstract}
We introduce generalized filtration with which we can represent situations
such as
some agents forget information at some specific time. 
The filtration is defined as a functor to a category
$\Prob$
%of probability spaces 
whose objects are
all probability spaces
and whose arrows correspond to measurable functions
 satisfying an absolutely continuous requirement
\cite{AR_2019}.
As an application of a generalized filtration,
we develop a binomial asset pricing model,
and
investigate the
valuations of financial claims along this type of non-standard filtrations.
\end{abstract}

% NOTE: it is good practice to label all headings (and proclamations) immediately

%%%%%%%%%%%%%%% Start Main Text %%%%%%%%%%%%%%%%%%

\section{Introduction}
\label{sec:intro}

It is well known that 
in stochastic process theory 
and theories 
developed on it
such as 
stochastic differential equation theory
and stochastic control theory,
the concept of 
\newword{filtration}
that expresses increasing information along time is important.
The idea that the world's information grows over time seems to be quite natural, 
but in a sense it is a divine perspective of omniscience and almighty, 
and
it would be a little different if we say that
 the amount of information that an individual has always increases with time. 
People forget and misunderstand.
The transition of such individuals' information may therefore be reduced, 
and may be remembered as a different form of experience than objective information.
The purpose of the first half of this paper is to propose a kind of 
\newword{subjective filtration}
that expresses the transition of such information.

In this way, 
we generalize the concept of filtration so that we can handle subjective situations, 
but the purpose of generalized filtration is not limited to that.
For example, 
consider a situation in which Black Swan, who no one had imagined up to a certain point in time, was falling.
The financial crisis that hit the world in 2008
and 
%the Chinese coronavirus infection 
the COVID-19 pandemic
in 2020 are typical examples.
When Black Swan suddenly appeared, which was not included among the possible future world lines, 
we could not give a probability for that event and we were greatly upset.
Of course, God could have a sufficiently large set of primitive events to take into account such possibilities,
it would have been possible to give a probability to an event that ordinary people did not expect.
But can such an idealized perspective really create a theory that averts the risk of Black Swan?

The generalized filtration formulated in this paper allows even the underlying set of probability space, 
which is the set of primitive events, to change over time.
And it allows the sudden appearance of Black Swan to be incorporated into the theory in a natural way.

In the second half of this paper, 
we consider two types of filtration on the binomial asset price model as an application of generalized filtration.
In particular, 
we show that there is a risk-neutral filtration associated with subjective filtration 
that a person who has lost memory for a certain period of time, 
and use it to price securities.
This indicates that people with a lack of memory can price securities.

%We also provide a complete form of a replication strategy making the valuation possible.

Finally, in summary, 
other applications of generalized filtration and future development directions are described.

\section{Generalized Filtrations}
\label{sec:genfil}

In this section, 
we define generalized filtration by gradually extending the classical filtration.

\subsection{Time Domains}
\label{timeDom}
A filtration represents a set of information that increases with time. 
The set of times here is called a 
\newword{time domain}
and is represented by
$\mathcal{T}$.
Typical
$\mathcal{T}$
has the following forms.
\begin{enumerate}
\item
$
\textred{\mathcal{T}}
	:=
\{
	0, 1, 2, \dots, T
\},
$

\item
$
\textred{\mathcal{T}}
	:=
\{
	0, 1, 2, \dots 
\},
$

\item
$
\textred{\mathcal{T}}
	:=
[0, T],
$

\item
$
\textred{\mathcal{T}}
	:=
[0, +\infty) ,
$
\end{enumerate}
where
$T$
is a time horizon.
In general
time domain may be a 
\newword{totally ordered set}
having the minimum element
$0$.

\subsection{Classical Filtrations}
\label{claFil}

Let
$
\textred{
\ProbO
}
	:= (
\Omega, \mathcal{F}, \mathbb{P}
)
$
be a probability space.
Let
$
\{ \textred{t_n} \}
$
be an increasing sequence
in a time domain
$\textblue{\mathcal{T}}$.
Then,
an increasing sequence
of
$\sigma$-fields
\begin{equation*}
\mathcal{F}_{t_0}
	\subset
\mathcal{F}_{t_1}
	\subset
\cdots
	\subset
\mathcal{F}_{t_n}
	\subset
\mathcal{F}_{t_{n+1}}
	\subset
\cdots
%\raisebox{0.5ex}[2ex][1ex]{\dots}
\end{equation*}
with
$
\mathcal{F}_{t_n}
	\subset
\mathcal{F}
$
is called a
\newword{classical filtration}.
In other words, a
\textblue{filtration}
is a family of set-inclusion relations like
\begin{equation*}
\{
\mathcal{F}_{s}
	\; \textblue{\subset} \;
\mathcal{F}_{t}
\}_{s \le t}  .
\end{equation*}
Now let
\[
\textred{
\ProbO_t
}
	:= (
\Omega, \textblue{\mathcal{F}_t}, \mathbb{P}
)
\]
be probability spaces
whose 
$\sigma$-fields are changing
per time
$t$.
Then,
for
$s \le t$
in
$\mathcal{T}$,
the condition that the function below defined as an identity function
$i_{s,t}$
is measurable 
is equivalent to the condition
$
\mathcal{F}_s
	\subset
\mathcal{F}_t
$
\[
\xymatrix@C=40 pt@R=25 pt{
	\ProbO_s
		\ar @{}_{ \mathrel{ \rotatebox[origin=c]{90}{$\in$} } } @<+2pt> [d]
&
	\ProbO_t
		\ar @{->}_{\textred{i_{s,t}}} [l]
		\ar @{}_{ \mathrel{ \rotatebox[origin=c]{90}{$\in$} } } @<+2pt> [d]
\\
	\omega
&
	\omega
		\ar @{|->} [l]
} .
\]
In other words,
the filtration can be identified with a family of measurable functions
\[
\xymatrix@C=30 pt@R=30 pt{
	\{ \ProbO_s
&
	\ProbO_t \}_{s \le t}
		\ar @{->}_{\textblue{i_{s,t}}} [l]
}.
\]
Therefore, in the following, 
instead of using the 
$\sigma$-field
$\mathcal{F}_t$,
filtration will be considered as a family of measurable functions.

\subsection{Generalization of Filtrations}
\label{sec:genFil}

As we see in the previous section,
a filtration can be seen as a family of
identity functions
$i_{s,t}$
as measurable functions.
Now what if we generalize them to arbitrary measurable functions like the following?
\[
\xymatrix@C=30 pt@R=30 pt{
	\{ \ProbO_s
&
	\ProbO_t \}_{s \le t}
		\ar @{->}_{\textred{f_{s,t}}} [l]
}
\]
satisfying
\begin{equation*}
f_{t,t}
	=
\textred{
Id_{\ProbO_t}
}
\quad \textrm{and} \quad
f_{s,t}
	\circ
f_{t,u}
	=
f_{s,u}
\end{equation*}
for any
$
s \le t \le u
$
in
$\mathcal{T}$,
where
$
Id_{\ProbO_t}
$
is an identity function on
$\ProbO_t$.
However,
this definition is too general
for a random variable
$
X
	:
\ProbO_t
	\to
\mathbb{R}
$
to define
its conditional expectation
$
\textblue{E^{f_{s,t}}(X)}
	:
\ProbO_s
	\to
\mathbb{R}
$
satisfying
\begin{equation*}
\int_A
	\textblue{E^{f_{s,t}}(X)}
d \mathbb{P}
	=
\int_{f_{s,t}^{-1}(A)}
	X
d \mathbb{P} 
\quad
(\forall A
	\in
\mathcal{F}_s
). 
\end{equation*}
In order to make it possible,
we need to add an extra condition to the measurable function
$f_{s,t}$
 called
\newword{null-preserving},
that is,
for any
$
A \in \mathcal{F}_s
$,
$
\mathbb{P}(A)
	= 0
$
implies
%\begin{equation*}
$
\mathbb{P}(f_{s,t}^{-1}(A))
	= 0 
$
\cite{adachi_2014crm}.
%\end{equation*}
In fact,
if
$f_{s,t}$
is null-preserving,
as we will see later,
we can define a conditional expectation
%\begin{equation*}
$
E^{f_{s,t}}(X)
	: \ProbO_s \to \mathbb{R}
$.
%\end{equation*}
%
Note that when the identity function is generalized to a null-preserving function, 
the corresponding sequence of the 
$\sigma$-fields is not necessarily monotonically increasing.

In order to give a further generalization,
%To further generalize, 
we consider that the probability space at each time fluctuates not only with 
the 
$\sigma$-fields
but also with probability measures and underlying sets.
In other words, 
the probability space
$
\ProbO_t
$
 at time 
$t$
is redefined as follows.
\begin{equation*}
\textred{
\ProbO_t
}
	:= (
\textred{\Omega_t}, \textblue{\mathcal{F}_t}, \textred{\mathbb{P}_t}
).
\end{equation*}
Along with this, 
the definition of null-preserving functions is extended as follows.

\begin{defn}
\label{defn:nullPres}
Let
$
\ProbO
	=
(\Omega, \mathcal{F}, \mathbb{P})
$
and
$
\ProbO'
	=
(\Omega', \mathcal{F}', \mathbb{P}')
$
be two probability spaces
and
$
f : 
\ProbO
	\to
\ProbO'
$
be a measurable functions between them.
Then 
$f$
is called
\newword{null-preserving}
if
\begin{equation*}
\mathbb{P} \circ f^{-1}
	\ll
\mathbb{P}'
\quad
(\textrm{absolutely continuous}) .
\end{equation*}
\end{defn}

\begin{defn}
\label{defn:genFil}
A
\newword{generalized filtration}
is a family of
null-preserving functions
%$f_{s,t}$
\[
\xymatrix@C=30 pt@R=30 pt{
	\{ \ProbO_s
&
	\ProbO_t \}_{s \le t}
		\ar @{->}_{\textred{f_{s,t}}} [l]
}
\]
satisfying
%\vspace{-0.8\baselineskip}
\begin{equation*}
f_{t,t}
	=
\textred{
Id_{\ProbO_t}
}
\quad \textrm{and} \quad
f_{s,t}
	\circ
f_{t,u}
	=
f_{s,u}
\end{equation*}
for
all triples
$
s \le t \le u
$
in
$\mathcal{T}$.
\end{defn}

Then, we obtain a following theorem.

%\begin{thm}{\normalfont{(Adachi-Ryu \cite{AR_2019})}}
\begin{thm}{\normalfont{(\cite{AR_2019})}}
For any random variable
$\textblue{X}$
on 
$\ProbO_t$
and any null-preserving function
$
f : 
\ProbO_t
	\to
\ProbO_s
$,
there exists a
random variable
$Y$
on
$\ProbO_s$
such that
for every 
$
\textblue{A}
	\in
\mathcal{F}_s
$,
\begin{equation}
\label{eq:condExp}
\int_A
	Y
	%\textblue{E^{f}(X)}
d \mathbb{P}_s
	=
\int_{f^{-1}(A)}
	X
d \mathbb{P}_t .
\end{equation}
We write 
$
\textblue{E^{f}(X)}
$
for
the random variable
$Y$,
and call it a
\newword{conditional expectation}
of
$X$
along
$f$.
\end{thm}
\begin{proof}
Define a measure
$X^*$
on
$
(\Omega_t, \mathcal{F}_t)
$
as in the following diagram.
\begin{equation*}
\xymatrix@C=15 pt@R=20 pt{
&&
   D
       \ar @{|->}^{} [rr]
       \ar @{}_{ \mathrel{ \rotatebox[origin=c]{-90}{$\in$} } } @<+6pt> [d]
&&
   \textred{X^*}(D)
       \ar @{}^-{:=} @<-6pt> [r]
       \ar @{}_{ \mathrel{ \rotatebox[origin=c]{-90}{$\in$} } } @<+6pt> [d]
&
   \int_D X \, d \mathbb{P}_t
\\
   \mathcal{F}_s
       \ar @{->}^{f^{-1}} [rr]
       \ar @/_2pc/_{\mathbb{P}_s} [rrrr]
&&
   \mathcal{F}_t
       \ar @{->}^{\textred{X^*}} @<+2pt> [rr]
       \ar @{->}_{\mathbb{P}_t} @<-2pt> [rr]
&&
   \mathbb{R}   
}
\end{equation*}
Then,
since
$
X^*
	\ll 
\mathbb{P}_t
$
and
$f$
is null-preserving,
we have
%\vspace{-0.8\baselineskip}
\begin{equation*}
X^* \circ f^{-1}
  \ll
\mathbb{P}_t \circ f^{-1}
  \ll
\mathbb{P}_s .
\end{equation*}
Therefore,
we get a following
Radon-Nikodym derivative.
\begin{equation*}
Y
  :=
%\frac{
  \partial
  ( X^* \circ f^{-1} )
/
%}{
  \partial
  \mathbb{P}_s .
%}
\end{equation*}
With this
$Y$
we obtain for every
$A \in \mathcal{F}_s$,
\begin{align*}
	&
\textblue{
 \int_A Y \, d \mathbb{P}_s
}
	=
\int_A \,
d (X^* \circ f^{-1})
%\\ &=
	=
(X^* \circ f^{-1})(A)
 =
X^* ( f^{-1}(A))
	=
\textblue{
 \int_{f^{-1}(A)} X \, d \mathbb{P}_t 
} .
\end{align*}

\end{proof}

Henceforth, 
generalized filtration will be referred to simply as filtration.

\subsection{Filtration is a Functor}
\label{sec:filfunc}

In this subsection, 
we will try to redefine the filtration introduced in 
Section \ref{sec:genFil}
using Category Theory
\cite{maclane1997}.

\begin{defn}{[Two Categories $\Prob$ and $\mathcal{T}$]}
\label{defn:ProbAndT}

\begin{enumerate}
\item
All probability spaces and all null-preserving functions between them form a
\newword{category}.
This category is denoted by
$\textred{\Prob}$
\footnote{
For further discussion about the category
$\textblue{\Prob}$,
see
\cite{AR_2019}.
}.

\item
A time domain
$\textblue{\mathcal{T}}$
can be regarded as a category
if we consider its elements as 
\newword{objects},
and
if two objects 
$s$
and
$t$
have one and only one
\newword{arrow}
from 
$t$
to
$s$
when there is a relation
$
s \le t
$.

\end{enumerate}

\end{defn}

%\def \XXX {0}
%\if  XXX
%\fi % XXX

Then,
the filtration introduced in
Section \ref{sec:genFil}
can be regarded as a
\newword{functor}
$
F : \mathcal{T} \to \Prob
$
(Figure\ref{fig:filfun}).
We sometimes call
$F$
a 
\newword{$\mathcal{T}$-filtration}
in order for clarifying its time domain.

%\begin{figure*}[tb]
\begin{figure*}[t]
\center
\[
\xymatrix@C=30 pt@R=30 pt{
   \mathcal{T}
      \ar @{->}^{\textred{F}} [rrrr]
&&&&
   \Prob
\\
   s
      \ar @(dl,ul)[]^{Id_s}
&&&&
	%F(s) := \ProbO_s
	F(s)
		\ar @(dl,ul)[]^{
			\textblue{
			F(Id_s) = Id_{F(s)}
			}
		}
		\ar @{}^-{:=} @<-6pt> [r]
&
	\textred{\ProbO_s}
\\
   t
      \ar @{->}^{\iota^N_{s,t}} [u]
      \ar @(dl,ul)[]^{Id_t}
&&&&
   F(t)
      \ar @{->}^{\textred{f_{s,t}} =: F(\iota^N_{s,t})} [u]
      \ar @(dl,ul)[]^{
        \textblue{
         F(Id_t) = Id_{F(t)}
        }
      }
\\
   u
      \ar @{->}^{\iota^N_{t,u}} [u]
      \ar @(dl,ul)[]^{Id_u}
      \ar @/_2pc/_{
         \iota^N_{s,t} \circ \iota^N_{t,u}
      } [uu]
&&&&
	%F(u) := \ProbO_u
	F(u)
		\ar @{->}^{\textred{f_{t,u}} =: F(\iota^N_{t,u})} [u]
		\ar @(dl,ul)[]^{
			\textblue{
			F(Id_u) = Id_{F(u)}
			}
		}
		\ar @/_2pc/_{
			\textblue{
			F(\iota^N_{s,t} \circ \iota^N_{t,u}) = F(\iota^N_{s,t}) \circ F(\iota^N_{t,u}) = \textred{f_{s,u}}
			}
		} [uu]
		\ar @{}^-{:=} @<-6pt> [r]
&
	\textred{\ProbO_u}
}
\]
\caption{Filtration
$
F : \mathcal{T} \to \Prob
$}
\label{fig:filfun}
\end{figure*}
%\end{table*}

%\begin{figure}[tb]
\begin{figure}[h]
\center
\textblue{Classical Filtration}
\begin{equation*}
\xymatrix@C=25 pt@R=15 pt{
\mathcal{T}
       %\ar @{->}_{\textred{F}} [d]
&
	t_0
		\ar @{}^-{\le} @<-6pt> [r]
&
	t_1
		\ar @{}^-{\le} @<-6pt> [r]
&
	t_2
		\ar @{}^-{\le} @<-6pt> [r]
&
	\cdots
\\
	\textred{\mathbb{F}}
&
	\textred{\mathcal{F}_{t_0}}
		\ar @{}^-{\subset} @<-6pt> [r]
&
	\textred{\mathcal{F}_{t_1}}
		\ar @{}^-{\subset} @<-6pt> [r]
&
	\textred{\mathcal{F}_{t_2}}
		\ar @{}^-{\subset} @<-6pt> [r]
&
	\cdots
\\
&
	\Omega
&
	\Omega
		\ar @{->}_{Id_{\Omega}} [l]
&
	\Omega
       \ar @{->}_{Id_{\Omega}} [l]
&
	\cdots
       \ar @{->}_{Id_{\Omega}} [l]
}
\end{equation*}

%\lspace

\textblue{Generalized Filtration}
\begin{equation*}
\xymatrix@C=25 pt@R=15 pt{
\mathcal{T}
       \ar @{->}_{\textred{F}} [d]
&
	t_0
&
	t_1
		\ar @{->} [l]
&
	t_2
		\ar @{->} [l]
&
	\cdots
		\ar @{->} [l]
\\
\textred{\Prob}
&
	\ProbO_{t_0}
&
	\ProbO_{t_1}
		\ar @{->}_{\textred{f_{t_0,t_1}}} [l]
&
	\ProbO_{t_2}
		\ar @{->}_{\textred{f_{t_1,t_2}}} [l]
&
	\cdots
		\ar @{->}_{\textred{f_{t_2,t_3}}} [l]
}
\end{equation*}
\caption{Classical and Generalized Filtrations}
\label{fig:distfil}

\end{figure}
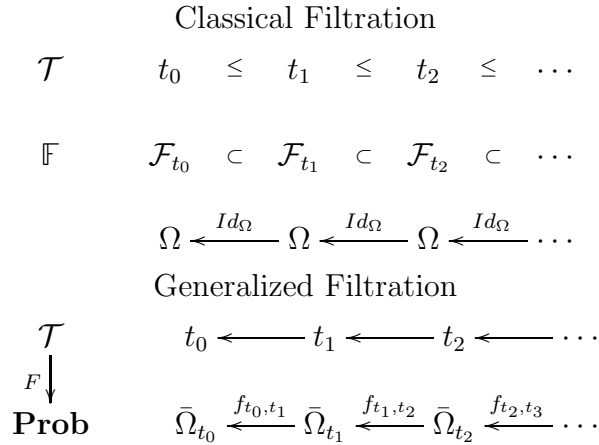

\section{Filtrations over a Binomial Asset Pricing Model}
\label{sec:BAPM}

In this section, 
as a concrete example of the filtration introduced in 
Section \ref{sec:genfil}, 
we look at an unusual filtration on a binomial asset pricing model.

\subsection{Filtration $\mathcal{B}^N$}
\label{sec:BMsetting}

First, we define a general scheme of our model by introducing a filtration
$\mathcal{B}^N$ 
for an integer
$N$.

\begin{defn}[{Time Domain and Probability Space}]
Let
$\textred{N} \in \mathbb{N}$
and
$
s, t \in \mathbb{R}
$
be non-negative real numbers.

\begin{enumerate}
\item
\newword{Discrete intervals}.
\begin{align*}
\textred{[s, t]^N}
	&:= \{
	n 2^{-N}
\mid
	n \in \mathbb{Z}
\; \mathrm{and} \;
	s \le n 2^{-N} \textblue{\le} t
\} ,
	\\
\textred{[s, t)^N}
	&:= \{
	n 2^{-N}
\mid
	n \in \mathbb{Z}
\; \mathrm{and} \;
	s \le n 2^{-N} \textblue{<} t
\},
	\\
\textred{(s, t]^N}
	&:= \{
	n 2^{-N}
\mid
	n \in \mathbb{Z}
\; \mathrm{and} \;
	s < n 2^{-N} \textblue{\le} t
\} ,
	\\
\textred{(s, t)^N}
	&:= \{
	n 2^{-N}
\mid
	n \in \mathbb{Z}
\; \mathrm{and} \;
	s < n 2^{-N} \textblue{<} t
\}.
\end{align*}

\item
Let
$\textred{\mathcal{T}^N}$
be a category
whose objects are elements of 
$[0, \infty)^N$.
For
$s, t \in [0, \infty)^N$,
$\textred{\mathcal{T}^N}$
has the unique arrow
$\textred{\iota^N_{s,t}}$
from $t$ to $s$
if and only if
$t \ge s$.

\item
$
\textred{B^N_t}
 :=
\{0,1\}^{(0, t]^N}
$.
\quad
(function space) 

\item
$
\textred{
\mathcal{F}^N_t
}
	:=
2^{B^N_t}
$.
\quad
(powerset)

\item
Let
$
\textred{p^N_s} 
	\in
%(0,1)
[0,1]
$
%be predefined constants 
for each
$
s \in (0, \infty)^N
$.
Then, a probability measure
$
\textred{
	\mathbb{P}^N_t
}
	:
\mathcal{F}^N_t
	\to
[0,1]
$
is defined 
for every
$\omega \in B^N_t$
by
% the following:
\begin{equation*}
	\textred{\mathbb{P}^N_t}
(\{\omega\})
	:=
\prod_{s \in (0,t]^N}
	(p^N_s)^{\omega(s)}
	(1 - p^N_s)^{1 - \omega(s)} .
\end{equation*}

\item
$
\textred{\ProbB}^N_t
	:=
(B^N_t, \mathcal{F}^N_t, \mathbb{P}^N_t)
$
\quad
(probability space)
%(\cite{shreve_I})
%\cite{CK2012DM}

\end{enumerate}

\end{defn}

\begin{defn}
\label{defn:filB}
A filtration
$\textred{\mathcal{B}^N}$
is determined by defining arrows
$\textred{f^N_{s,t}}$
below:
\begin{equation*}
\xymatrix@C=40 pt@R=8 pt{
	\mathcal{T}^N
		\ar @{->}^{\textred{\mathcal{B}^N}} [r]
&
	\Prob
\\
	s
&
	\ProbB^N_s
\\\\
	t
		\ar @{->}^{\iota^N_{s,t}} [uu]
&
	\ProbB^N_t
		\ar @{->}_{
			%\textred{f^N_{s,t}} := \mathcal{B}^N(\iota^N_{s,t})
			\mathcal{B}^N(\iota^N_{s,t}) := f^N_{s,t} 
		} [uu]
}
\end{equation*}

The filtration
$\mathcal{B}^N$
is called
\newword{non-trivial}
if there exists
$
t \in (0, \infty)^N
$
such that
$
0 < p_t < 1
$.

\end{defn}

Note that
for a non-trivial filtration
$\mathcal{B}^N$,
every function
from
$B^N_t$
to
$B^N_s$
becomes a 
null-preserving function
from
$\ProbB^N_t$
to
$\ProbB^N_s$.

\lspace

As we introduced,
the functor
$
\textred{
	\mathcal{B}^N
}
$
is a
generalized filtration,
representing a filtration over
the classical binomial model
 developed,
for example
 in
\cite{shreve_I}.

The classical version requires the terminal time horizon
$\textred{T}$
for determining the underlying set
$
\textred{\Omega}
	:=
\{0, 1\}^T
$
while our version does not require it
since the time variant probability spaces can evolve without any limit.
That is, 
\textblue{
our version allows unknown future elementary events,
which, we believe, shows a big philosophical difference from the traditional Kolmogorov world.
}

\begin{prop}
\label{prop:BMcondE}
For
a random variable
$X$
on
$\bar{B}^N_t$
and
$
\textred{\omega}
	\in
\bar{B}^N_s
$,
we have
\begin{equation*}
%\label{eq:ceBin}
\textred{
	E^{f^N_{s,t}}(X)
}(\omega)
\mathbb{P}^N_s(\{\omega\})
	=
\sum_{\omega' \in \textblue{(f^N_{s,t})^{-1}(\omega)}}
	X(\omega') 
	\mathbb{P}^N_{t}(\{\omega'\}) .
\end{equation*}

\end{prop}

\begin{proof}
Put
$
A := \{ \omega \}
$
and
$
f_{s, t} := f^N_{s, t}
$
in
(\ref{eq:condExp}).
Then the result is straightforward.
\end{proof}

In order to see a variety of filtrations, we introduce two candidates of
$f^N_{s,t}$
introduced in
Definition \ref{defn:filB}.

\begin{defn}[{Two Candidates of $f^N_{s,t}$}]
\label{defn:twoCandF}
Let
$
s, t
$
be objects of
$\mathcal{T}^N$
satisfying
$
s < t
$.

\begin{enumerate}
\item
$
\textred{
	\mathbf{full}^N_{s,t}
}
$
%\vspace{-0.8\baselineskip}
\begin{equation*}
\xymatrix@C=35 pt@R=15 pt{
	\ProbB^N_s
        \ar @{}_{ \mathrel{ \rotatebox[origin=c]{90}{$\in$} } } @<+6pt> [d]
&
	\ProbB^N_t
        \ar @{}_{ \mathrel{ \rotatebox[origin=c]{90}{$\in$} } } @<+6pt> [d]
        \ar @{->}_{\textred{
			\mathbf{full}^N_{s,t}
		}} [l]
\\
	\omega \mid_{(0, s]^N}
&
	\omega
        %\ar @{|->}_{\textred{\mathbf{full}^N_{s,t}}} [l]
        \ar @{|->} [l]
}
\end{equation*}
 
%\vspace{-0.8\baselineskip}
\item
$
\textred{
	\mathbf{drop}^N_{s,t}
}
$
\begin{equation*}
\xymatrix@C=35 pt@R=15 pt{
	\ProbB^N_s
        \ar @{}_{ \mathrel{ \rotatebox[origin=c]{90}{$\in$} } } @<+6pt> [d]
&
	\ProbB^N_t
        \ar @{}_{ \mathrel{ \rotatebox[origin=c]{90}{$\in$} } } @<+6pt> [d]
        \ar @{->}_{\textred{
			\mathbf{drop}^N_{s,t}
		}} [l]
\\
	\mathbf{full}^N_{s,t}(\omega) \times \textblue{\mathbb{1}_{(0,s\textred{)}^N}}
&
	\omega
        %\ar @{|->}_{\textred{\mathbf{drop}^N_{s,t}}} [l]
        \ar @{|->} [l]
}
\end{equation*}
 
\end{enumerate}
\end{defn}

The function 
$\mathbf{drop}^N_{s,t}$
can be interpreted to forget what happens at time
$s$.

%\lspace

We can easily show the following proposition.

\begin{prop}
\label{prop:fulldrop}
For
$
s < t < u
$
in
$
[0, \infty]^N
$,
\begin{enumerate}
\item
$
\mathbf{full}^N_{s,t}
	\circ
\mathbf{full}^N_{t,u}
	=
\mathbf{full}^N_{s,u}
$ ,

\item
$
\mathbf{full}^N_{s,t}
	\circ
\mathbf{drop}^N_{t,u}
	=
\mathbf{full}^N_{s,u}
$ ,

\item
$
\mathbf{drop}^N_{s,t}
	\circ
\mathbf{full}^N_{t,u}
	=
\mathbf{drop}^N_{s,u}
$ ,

\item
$
\mathbf{drop}^N_{s,t}
	\circ
\mathbf{drop}^N_{t,u}
	=
\mathbf{drop}^N_{s,u}
$ .
\end{enumerate}
\end{prop}

\begin{defn}[{Examples of (Subjective) Filtrations}]
\label{defn:expSurFil}
Let
$
s, t
$
be
any objects of
$\mathcal{T}^N$
such that
$
s < t
$.
\begin{enumerate}
\item
\textblue{Classical filtration}:
$
\textred{
	\mathbf{Full}^N
}
	: \mathcal{T}^N \to \Prob
$
is defined by
\begin{equation*}
\textred{
	\mathbf{Full}^N(
		\iota^N_{s,t}
	)
}
	:=
\textblue{
	\mathbf{full}^N_{s,t}
} .
\end{equation*}

%\lspace
\item
\textblue{Dropped filtration}:
$
\textred{
	\mathbf{Drop}^N_{\alpha, \beta}
}
	: \mathcal{T}^N \to \Prob
$
where
$
\textblue{\alpha},
\textblue{\beta}
	\in
\mathbb{R}
$
are constants,
is defined by
\begin{equation*}
\textred{
	\mathbf{Drop}^N_{\alpha, \beta}
	(
		\iota^N_{s,t}
	)
}
	:= \begin{cases}
\textblue{\mathbf{drop}^N_{s,t}}
\quad \textrm{if} \quad
s \ne t
\; \textrm{and} \;
s \in [\textblue{\alpha}, \textblue{\beta}]^N ,
	\\
\textblue{\mathbf{full}^N_{s,t}}
\quad \textrm{otherwise} .
\end{cases}
\end{equation*}
A person who has a subjective filtration
$\mathbf{Drop}^N_{\alpha, \beta}$
forgets the events happened during 
$[\alpha, \beta]$.
\end{enumerate}

\end{defn}

Note also that the dropped filtration is \textblue{well-defined} by
Proposition \ref{prop:fulldrop}.

\begin{defn}{[$\mathcal{B}$-Adapted Process $\xi^N_t$]}
\label{defn:adaptedPrs}
Let
$
t \in [0,+\infty)^N
$.
a stochastic process
$
\textred{\xi^N_t}
	:
B^N_t \to \mathbb{R}
$
is defined by
\begin{equation*}
\xi^N_t(\omega)
	:=
2 \omega(t) - 1
	\quad
(\forall \omega \in B^N_t) .
\end{equation*}

\end{defn}

\begin{defn}
\label{defn:Itj}
For
$j = 0, 1$
and
$
\omega
	\in
B^N_t
$,
\begin{equation*}
I^N_{t}(j,\omega)
	:=
\{
	\omega'
		\in
	(f^N_{t,t+2^{-N}})^{-1}(\omega)
		\mid
	\omega'(t+2^{-N}) = j
\} .
\end{equation*}
\end{defn}

\begin{prop}
\label{prop:xiPr}
For
$
\omega \in B^N_t
$
with
$
\mathbb{P}^N_t(\omega) \ne 0
$,
\begin{equation*}
E^{
	f^N_{t,t+2^{-N}}
}
(\xi^N_{t+2^{-N}})
(\omega)
	=
\#((f^N_{t,t+2^{-N}})^{-1}(\omega)) p^N_{t+2^{-N}}
	- 
\#I^N_{t}(0, \omega) ,
\end{equation*}
where
$
\textred{\#}A
$
denotes the cardinality of the set 
$A$.
\end{prop}
\begin{proof}
By Proposition \ref{prop:BMcondE},
\begin{align*}
E^{
	f^N_{t,t+2^{-N}}
}
(\xi^N_{t+2^{-N}})
(\omega)
	&=
\sum_{\omega' \in 
	(f^N_{t,t+2^{-N}})^{-1}(\omega)
	}
	\frac{
		\xi^N_{t+2^{-N}} (\omega')
		\mathbb{P}^N_{t+2^{-N}}(\omega')
	}{
		\mathbb{P}^N_t(\omega)
	}
	\\&=
\sum_{\omega' \in 
	(f^N_{t,t+2^{-N}})^{-1}(\omega)
	}
	\frac{
		(2 \omega'(t+2^{-N}) - 1)
		\mathbb{P}^N_{t+2^{-N}}(\omega')
	}{
		\mathbb{P}^N_t(\omega)
	}
	\\&=
\sum_{\omega' \in I^N_t(1, \omega)}
	\frac{\mathbb{P}^N_{t+2^{-N}}(\omega')}{\mathbb{P}^N_t(\omega)}
	-
\sum_{\omega' \in I^N_t(0, \omega)}
	\frac{\mathbb{P}^N_{t+2^{-N}}(\omega')}{\mathbb{P}^N_t(\omega)}
	\\&=
\sum_{\omega' \in I^N_t(1, \omega)}
	p^N_{t+2^{-N}}
	-
\sum_{\omega' \in I^N_t(0, \omega)}
	(1 - p^N_{t+2^{-N}})
	\\&=
\#((f^N_{t,t+2^{-N}})^{-1}(\omega)) p^N_{t+2^{-N}}
	- 
\#I^N_{t}(0, \omega) .
\end{align*}
\end{proof}

\subsection{Arbitrage Strategies}
\label{sec:ArbStr}

Now we define two instruments tradable in our market.

\begin{defn}{[Stock and Bond Processes]}
\label{defn:stockAndBond}
Let
$
\mu, \sigma, r
	\in
\mathbb{R}
$
be constants
such that
$
\sigma > 0
$,
$
\mu > \sigma - 1
$
and
$
r > -1
$.
%For
%$\textred{\mu}, \textred{r} \in \mathbb{R}$
%and
%$\textred{\sigma} \in \mathbb{R}_+$,
We have the following 
$\mathcal{B}^N$-adapted processes
which are
two instruments tradable in our market.
Let
$
t \in [0,+\infty)^N
$.
%and
%$ \omega \in B^N_t $.
\begin{enumerate}
\item
A 
\newword{stock process} 
$
\textred{S^N_t}
	:
B^N_t \to \mathbb{R}
$
 over the filtration 
$\mathcal{B}^N$
is defined by
\begin{equation*}
S^N_0(\textred{*})
	:= \textred{s_0},
\quad
S^N_{t+2^{-N}}
	:=
(S^N_t \circ \textblue{f^N_{t,t+2^{-N}}})
(1 + 2^{-N} \mu + 2^{-\frac{N}{2}} \sigma \xi^N_{t+2^{-N}})
\end{equation*}
where
$
\textred{*}
	\in
B^N_0
$
is the unique element.

\item
A \newword{bond process}
$
\textred{b^N_t}
	:
B^N_t \to \mathbb{R}
$
 over the filtration 
$\mathcal{B}^N$
is defined by
\begin{equation*}
b^N_0(*) := 1,
\quad
b^N_{t+2^{-N}}
	:=
(b^N_t \circ \textblue{f^N_{t,t+2^{-N}})}
(1+2^{-N}r) .
\end{equation*}

\end{enumerate}

\end{defn}

The condition
$
\mu > \sigma - 1
$
is necessary for keeping the stock price positive.

We sometimes call the triple
$
(\mathcal{B}^N, S^N, b^N)
$
a
\newword{market}.
But, it does not mean that
the market will not contain other instruments.

The following proposition is straightforward.

\begin{prop}
\label{prop:EfNStNxiN}
Let
$
1_{B^N_t}
$
be a random variable 
on
$
B^N_t
$
defined by
$
1_{B^N_t}
(\omega)
	= 1
$
for every
$\omega \in B^N_t$.
Then, we have
for any
$
\omega \in B^N_t
$,
\begin{enumerate}
\item
%\begin{equation}
$
E^{f^N_{t,t+2^{-N}}}
(S^N_{t+2^{-N}})
	=
S^N_t
\big(
(1 + 2^{-N}\mu)
E^{f^N_{t,t+2^{-N}}}
(1_{B^N_{t+2^{-N}}})
	+
2^{-\frac{N}{2}} \sigma
E^{f^N_{t,t+2^{-N}}}
{\xi^N_{t+2^{-N}}})
\big) ,
$
%\end{equation}

\item
%\begin{equation}
$
E^{f^N_{t,t+2^{-N}}}
(1_{B^N_{t+2^{-N}}})
	(\omega)
	=
\{
	\mathbb{P}^N_{t+2^{-N}}
	((f^N_{t,t+2^{-N}})^{-1}(\omega))
	%(f_n^{-1}(a))
\}
	/
\{
	\mathbb{P}^N_{t}
	(\omega)
\} ,
$
%\end{equation}

\item
%\begin{equation}
$
b^N_t(\omega)
	=
(1+2^{-N}r)^{2^N t} .
$
%\end{equation}
\end{enumerate}
\end{prop}

\begin{defn}{[Strategies]}
A
\newword{strategy}
is a sequence
$
(\phi, \psi)
	=
\{
(\phi_t, \psi_t)
\}_{t \in (0, \infty)^N }
$,
where
\begin{equation}
\phi_t : B^N_{t-2^{-N}} \to \mathbb{R}
\; \textrm{and} \;
\psi_t : B^N_{t-2^{-N}} \to \mathbb{R} .
\end{equation}
Each element of the strategy
$
(\phi_t, \psi_t)
$
is called a 
\newword{portfolio}.
For
$
t \in [0, \infty)^N
$,
the 
\newword{value}
$V_t$
 of the portfolio at time $t$
is determined by:
\begin{equation}
V_t
	:=
\begin{cases}
S^N_0 \phi_{2^{-N}}
	+
b^N_0 \psi_{2^{-N}}
	\quad \textrm{if} \quad
t = 0 ,
	\\
S^N_t 
( \phi_{t} \circ 
	f^N_{t-2^{-N},t}
)
	+
b_n 
( \psi_{t} \circ
	f^N_{t-2^{-N},t}
)
	\quad \textrm{if} \quad
t > 0 .
\end{cases} 
\end{equation}

\end{defn}

\begin{defn}{[Gain Processes]}
A
\newword{
gain process
}
of the strategy
$
(\phi, \psi)
$
is the process
$
\{
G^{(\phi, \psi)}_t
\}_{t \in [0, \infty)^N }
$
defined by
\begin{equation}
G^{(\phi, \psi)}_t
	:=
\begin{cases}
	-
(
	S^N_0 \phi_{2^{-N}}
		+
	b^N_0 \psi_{2^{-N}}
)
\; \textrm{if} \;
t = 0 ,
\\
(
	S^N_t (\phi_t \circ 
		f^N_{t-2^{-N},t}
	)
		+
	b^N_t (\psi_t \circ 
		f^N_{t-2^{-N},t}
	)
)
	-
(
	S^N_t \phi_{t+2^{-N}}
		+
	b^N_t \psi_{t+2^{-N}}
)
\; \textrm{if} \;
t > 0 .
\end{cases}
\end{equation}
\end{defn}

\begin{lem}
\label{lem:spbp}
Let 
$
t
\in [0, \infty)^N
$
with
\begin{equation}
S^N_t \phi_{t+2^{-N}}
	+
b^N_t \psi_{n+2^{-N}}
	= 0 .
\end{equation}
Then, we have
\begin{equation}
S^N_{t+2^{-N}}
(\phi_{t+2^{-N}} \circ 
	f^N_{t, t+2^{-N}}
)
	+
b_{n+1}
(\psi_{n+2^{-N}} \circ 
	f^N_{t, t+2^{-N}}
)
	=
(2^{-N} \mu + 2^{-\frac{N}{2}} \sigma \xi^N_{t+2^{-N}} - 2^{-N} r)
( (S^N_t \phi_{t+2^{-N}}) \circ 
	f^N_{t, t+2^{-N}}
) .
\end{equation}

\end{lem}
\begin{proof}
\begin{align*}
%& S_{n+1}
%(\phi_{n+1} \circ f_n) + b_{n+1} (\psi_{n+1} \circ f_n)
	&
LHS
	\\=&
(S^N_t \circ 
	f^N_{t, t+2^{-N}}
)
(1 + 2^{-N} \mu + 2^{-\frac{N}{2}} \sigma \xi^N_{t+2^{-N}})
(\phi_{t+2^{-N}} \circ 
	f^N_{t, t+2^{-N}}
)
	\\&
	+
(b^N_t \circ 
	f^N_{t, t+2^{-N}}
)
(1 + 2^{-N} r)
(\psi_{t+2^{-N}} \circ 
	f^N_{t, t+2^{-N}}
)
	\\=&
(1 + 2^{-N} \mu + 2^{-\frac{N}{2}} \sigma \xi^N_{t+2^{-N}})
((S^N_t \phi_{t+2^{-N}})
	\circ 
	f^N_{t, t+2^{-N}}
)
	+
(1 + 2^{-N} r)
((b^N_t \psi_{t+2^{-N}}) \circ 
	f^N_{t, t+2^{-N}}
)
	\\=&
(1 + 2^{-N} \mu + 2^{-\frac{N}{2}} \sigma \xi^N_{t+2^{-N}})
((S^N_t \phi_{t+2^{-N}})
	\circ 
	f^N_{t, t+2^{-N}}
)
	-
(1 + 2^{-N} r)
((S^N_t \phi_{t+2^{-N}}) \circ 
	f^N_{t, t+2^{-N}}
)
	\\=&
RHS .
\end{align*}
\end{proof}

\begin{defn}{[Arbitrage Strategies]}
\begin{enumerate}
\item
A strategy
$
(\phi, \psi)
$
is called a
$\mathcal{B}^N$-\newword{arbitrage strategy}
if
$
\mathbb{P}^N_{t}\big(
	G^{(\phi, \psi)}_t
		\ge 0
\big) = 1
$
for every 
$
t \in [0, \infty)^N
$,
and
$
\mathbb{P}^N_{t_0}\big(
	G^{(\phi, \psi)}_{t_0}
	> 0
\big) > 0
$
for some
$
t_0 \in [0, \infty)^N
$.

\item
The market is called 
\newword{non-arbitrage}
or
\textbf{NA}
if it does not allow 
$\mathcal{B}^N$-arbitrage strategies.

\end{enumerate}
\end{defn}

\begin{prop}
If
the market
$
(\mathcal{B}^N, S^N, b^N)
$
with a non-trivial filtration
$\mathcal{B}^N$
 is non-arbitrage,
then
%\begin{equation}
$
\abs{
	\textred{\mu}
		-
	\textred{r}
}
	< 
2^{\frac{N}{2}}
\textred{\sigma} .
%\end{equation}
$

\end{prop}
\begin{proof}
Assuming that
$
r 
	\le 
\mu
	- 
2^{\frac{N}{2}} \sigma
$
or
$
r 
	\ge 
\mu
	- 
2^{\frac{N}{2}} \sigma
$,
we will construct an arbitrage strategy
$
(\phi, \psi)
$
by using the following algorithm.
\begin{lstlisting}[basicstyle=\ttfamily\footnotesize,frame=single]
for t = 0, 1, 2, ...:
    t := 2^(-N) n
    observe S(t) and b(t)
    if r <= mu - 2^N  sigma:
        phi(t+2^(-N)) > 0  # pick arbitrarily
    elsif r >= mu + 2^N  sigma:
        phi(t+2^(-N)) < 0  # pick arbitrarily
    psi(t+2^(-N)) := -(S(t) / b(t)) phi(t+2^(-N))
    # Then, we have S(t) phi(t+2^(-N)) + b(t) psi(t+2^(-N)) = 0,
    # which simplifies the computation of G(t) in the following.
    if t == 0:
        G(0) := 0
    else:  # t > 1
        G(t) := S(t) (phi(t) * f(t-2^(-N))) + b(t) (psi(t) * f(t-2^(-N)))
\end{lstlisting}
In the above code, 
`*'
is the function composition operator.

By Lemma \ref{lem:spbp},
we have
\begin{equation}
G^{(\phi, \psi)}_{t}
	=
2^{-N}
(\mu
	+ 
2^{\frac{N}{2}} \sigma \xi^N_t
	-
r)
((S^N_{t-2^{-N}} \phi_t) \circ 
	f^N_{t-2^{-N}, t}
) .
\end{equation}
So
we have
$
G^{(\phi, \psi)}_{t}
\ge 0
$
as long as 
$
r \le \mu - 2^{N} \sigma
$
or
$
r \ge \mu + 2^{N} \sigma
$.

By the way,
since our filtration is non-trivial,
there exists a number 
$t_0$
such that
$
0 < p_{t_0} < 1
$.
It is easy to check that
\begin{equation}
\mathbb{P}^N_{t_0}(
	G^{(\phi, \psi)}_{t_0}
	>0
) > 0 ,
\end{equation}
which concludes that
$
(\phi, \psi)
$
is an arbitrage strategy.
\end{proof}

\subsection{Risk-Neutral Filtrations}
\label{sec:riskNeuFil}

In this subsection,
we assume that
$
\abs{
	\textred{\mu}
		-
	\textred{r}
}
	< 
2^{\frac{N}{2}}
\textred{\sigma}
$.

Let us consider about the following discounted stock process

\begin{defn}
\label{defn:discountStock}
A \newword{discount stock process}
$
\textred{(S^N_t)'}
	:
B^N_t \to \mathbb{R}
$
is defined by
\begin{equation*}
(S^N_t)'
	:=
(b^N_t)^{-1}
S^N_t .
\end{equation*}
\end{defn}

\begin{defn}
\label{defn:riskNeutralFil}
A
\newword{risk-neutral filtration}
with respect to the filtration
$\mathcal{B}^N$
is a filtration
$\textred{\mathcal{C}^N}$
such that
\vspace{-0.8\baselineskip}
\begin{equation*}
U \circ \mathcal{C}^N
	=
U \circ \mathcal{B}^N ,
\end{equation*}
where 
%\begin{equation*}
$
\textred{U} :
\Prob
	\to
\mathbf{Meas}
$
%\end{equation*}
is the \textblue{forgetful functor}
to the category of measurable spaces,
\vspace{-0.8\baselineskip}
\[
\xymatrix{
	\mathcal{T}^N
        \ar @{->}^{\textred{\mathcal{C}^N}} @<+3pt> [r]
        \ar @{->}_{\mathcal{B}^N} @<-3pt> [r]
&
    \Prob
        \ar @{->}^{\textred{U}} [r]
&
	\mathbf{Meas}
}
\vspace{-0.8\baselineskip}
\]
and with which
$
\textblue{(S^N_t)'}
$
becomes a 
\textblue{$\mathcal{C}^N$-martingale},
that is,
\begin{equation*}
E^{
	\textblue{\mathcal{C}^N(\iota^N_{s,t})}
}
(
	(S^N_t)'
)
	=
	(S^N_s)' .
\end{equation*}

\end{defn}

In the remainder of this subsection, 
we will focus on proving the following theorem.

%\begin{thm}[{\cite{ANR_2019_1}}]
\begin{thm}
There exists a risk-neutral filtration
with respect to
the 
filtration
$
\mathbf{Drop}^N_{\alpha, \beta}
$.
\end{thm}

First,
we examine what form the probability measure
$
\mathbb{Q}^N_t
	:
\mathcal{F}^N_t
	\to
[0,1]
$
takes
when
$
\mathcal{C}^N(t)
	=
(
B^N_t,
\mathcal{F}^N_t,
\mathbb{Q}^N_t
)
$
for a risk-neutral filtration
$\mathcal{C}^N$,
in general.

\begin{thm}
\label{thm:iiffMart}
A stochastic process
$(S^N_t)'$
is a
\textblue{
$\mathcal{C}^N$-martingale
}
if and only if
the following equation holds
for every
$
t \in [0,\infty)^N
$
and
$
\omega \in B^N_t
$.
\begin{equation*}
%\label{eq:thmDSP}
\mathbb{Q}^N_t(
	\{\omega\}
)
	=
c_1 \, \mathbb{Q}^N_{t+2^{-N}}(I^N_{t}(1,\omega))
	+
c_0 \, \mathbb{Q}^N_{t+2^{-N}}(I^N_{t}(0,\omega))
\end{equation*}
where
\begin{equation*}
\textred{c_1}
	:=
\frac{1 + 2^{-N} \mu + 2^{-\frac{N}{2}} \sigma}{1 + 2^{-N} r} ,
\quad
\textred{c_0}
	:=
\frac{1 + 2^{-N} \mu - 2^{-\frac{N}{2}} \sigma}{1 + 2^{-N} r} .
\end{equation*}

\end{thm}
\begin{proof}
Let
$\omega \in B^N_t$.
Then,
by Proposition \ref{prop:BMcondE},
we have
\begin{align*}
&
E^{\mathcal{C}^N(\iota^N_{t,t+2^{-N}})}(
	(S^N_{t+2^{-N}})'
)
(\omega)
\mathbb{Q}^N_{t}(\{\omega\})
	\\=&
\sum_{
	\omega' \in
	(f^N_{t,t+2^{-N}})^{-1}(\omega)
}
	(S^N_{t+2^{-N}})'(\omega')
	\mathbb{Q}^N_{t+2^{-N}}(\{\omega'\})
	\\=&
\sum_{
	\omega' \in
	(f^N_{t,t+2^{-N}})^{-1}(\omega)
}
	(b^N_{t+2^{-N}})^{-1}(\omega')
	(S^N_t \circ f^N_{t,t+2^{-N}}(\omega')
%\\&
	(1 + 2^{-N}\mu + 2^{-\frac{N}{2}}\sigma \xi^N_{t+2^{-N}}(\omega'))
	\mathbb{Q}^N_{t+2^{-N}}(\{\omega'\})
	\\=&
\sum_{
	\omega' \in
	(f^N_{t,t+2^{-N}})^{-1}(\omega)
}
	(1+ 2^{-N}r)^{-(t+2^{-N})}
	S^N_{t} (\omega)
%\\&
	(1 + 2^{-N}\mu + 2^{-\frac{N}{2}}\sigma \xi^N_{t+2^{-N}}(\omega'))
	\mathbb{Q}^N_{t+2^{-N}}(\{\omega'\})
	\\=&
(S^N_t)'(\omega)
\sum_{
	\omega' \in
	(f^N_{t,t+2^{-N}})^{-1}(\omega)
}
	\frac{
		1 + 2^{-N}\mu + 2^{-\frac{N}{2}}\sigma \xi^N_{t+2^{-N}}(\omega')
	}{
		1+ 2^{-N}r
	}
	%\\&
	\mathbb{Q}^N_{t+2^{-N}}(\{\omega'\}) .
\end{align*}
\noindent
Therefore,
the condition
%\begin{equation*}
$
(S^N_t)'
	=
E^{\mathcal{C}^N(\iota^N_{t,t+2^{-N}})}(
	(S^N_{t+2^{-N}})'
)
$
%\end{equation*}
is equivalent to
\begin{align*}
\mathbb{Q}^N_t(\{\omega\})
	&=
\sum_{
	\omega' \in I^N_t(1, \omega)
}
	\!\!
	\frac{
		1 + 2^{-N}\mu + 2^{-\frac{N}{2}}\sigma
	}{
		1+ 2^{-N}r
	}
	\mathbb{Q}^N_{t+2^{-N}}(\{\omega'\})
	%\\&
+
\sum_{
	\omega' \in I^N_t(0, \omega)
}
	\!\!
	\frac{
		1 + 2^{-N}\mu - 2^{-\frac{N}{2}}\sigma
	}{
		1+ 2^{-N}r
	}
	\mathbb{Q}^N_{t+2^{-N}}(\{\omega'\})
	\\&=
c_1 \,
	\mathbb{Q}^N_{t+2^{-N}}
	(I^N_t(1, \omega))
+
c_0 \,
	\mathbb{Q}^N_{t+2^{-N}}
	(I^N_t(0, \omega)) .
\end{align*}

\end{proof}

\begin{defn}
For
$
\omega
	\in
B^N_t
$
and
$
d \in \{0, 1\}
$,
$
(\omega d)
	\in
B^N_{t+2^{-N}}
$
is an element
satisfying
\begin{equation*}
(\omega d) (s)
	:= \begin{cases}
\omega(s)
	\quad &
%\textrm{if} \quad
(
s \le t
)
\\
d
	\quad &
%\textrm{if} \quad
(
s = t+2^{-N}
)
\end{cases}
\end{equation*}
for any
$
s \in
(0, t+2^{-N}]^N
$.

Unless there is confusion,
we will omit the parentheses in
$
((\omega d_1) d_2)
$
and write
$
\omega d_1 d_2
$.
\end{defn}

In order to determine more detail of
$\mathcal{C}$,
we need the following condition for
$\mathbb{Q}^N_t$.

\begin{prop}
\label{prop:Qcond}
The following conditions for
$
\mathbb{Q}^N_t
$
are equivalent.
\begin{enumerate}
\item
For all
$
t \in [0,\infty)^N
$
and
$
\omega \in B^N_t
$,
\begin{equation}
\mathbb{Q}^N_{t+2^{-N}}(\{\omega 0, \omega 1\})
	=
\mathbb{Q}^N_t(\{\omega\}) .
\end{equation}

\item
For all
$
t \in [0,\infty)^N
$,
$\mathbf{full}^N_{t,t+2^{-N}}$
is measure-preserving
w.r.t.
$\mathbb{Q}^N_t$,
that is,
\begin{equation}
\mathbb{Q}^N_{t}
	=
\mathbb{Q}^N_{t+2^{-N}}
	\circ
(\mathbf{full}^N_{t,t+2^{-N}})^{-1} .
\end{equation}

\item
There exists a sequence of functions
$
\{
	q_t
		:
	B^N_t \to [0,1]
\}_{t \in (0,\infty)^N}
$
satisfying
the following conditions 
for every
$
t \in (0, \infty)^N
$
and
$
\omega
	\in
B^N_t
$,
	\begin{enumerate}
	\item
$
\mathbb{Q}^N_t
(\{\omega\})
	=
\prod_{
	s \in (0, t]^N
}
	q_s(\omega \mid_{(0,s]^N}) ,
$

	\item
$
	q_{t+2^{-N}} (\omega 0)
+
	q_{t+2^{-N}} (\omega 1)
=
	1 .
$
	\end{enumerate}

\end{enumerate}
\end{prop}

In the following discussion, 
we assume the following assumption
which is the condition
%(2) of Proposition \ref{prop:Qcond}.
(3) of Proposition \ref{prop:Qcond}.

%\begin{asm}
%\label{asm:asmQ}
%%\label{asm:fullmp}
%For all
%$
%t \in [0,\infty)^N
%$,
%$\mathbf{full}^N_{t,t+2^{-N}}$
%is measure-preserving
%w.r.t.
%$\mathbb{Q}^N_t$.
%\end{asm}

\begin{asm}
\label{asm:asmQ}
Suppose that
there exists a sequence of functions
$
\{
	q_t
		:
	B^N_t \to [0,1]
\}_{t \in (0,\infty)^N}
$
satisfying
the following conditions 
for every
$
t \in (0, \infty)^N
$
and
$
\omega
	\in
B^N_t
$,
\begin{enumerate}
\item
$
\mathbb{Q}^N_t
(\{\omega\})
	=
\prod_{
	s \in (0, t]^N
}
	q_s(\omega \mid_{(0,s]^N}) ,
$

\item
$
	q_{t+2^{-N}} (\omega 0)
+
	q_{t+2^{-N}} (\omega 1)
=
	1 .
$
%$
%\sum_{
%	\omega'
%		\in
%	\{
%	\omega'
%		\in
%	B_{t+2^{-N}}
%\,
%	\mid
%\,
%	\omega'|_{(0,t]^N} = \omega
%	\}
%}
%	q_{t+2^{-N}}
%	(\omega')
%=
%	1
%$.

\end{enumerate}

\end{asm}

In the rest of this section,
we assume 
Assumption \ref{asm:asmQ},
and then
will determine
the risk-neutral filtration
$\mathcal{C}^N$
by calculating
$
\{
	q_t
\}_{t \in (0,\infty)^N}
$.

\begin{lem}
\label{lem:qbin}
Let
$c_1$
and
$c_0$
are constants defined in
Theorem \ref{thm:iiffMart}.
Then for any
$
x
	\in
\mathbb{R}
$
we have
\[
1
	=
c_1 x + c_0 (1-x)
		\quad \Longleftrightarrow \quad
x
	=
\frac{1}{2}
	-
2^{\frac{N}{2}-1}
\frac{
	\mu - r
}{
	\sigma
}
\quad \textrm{and} \quad
1 - x
	=
\frac{1}{2}
	+
2^{\frac{N}{2}-1}
\frac{
	\mu - r
}{
	\sigma
} .
\]
\end{lem}

%\subsection{$\mathfrak{f}^{\textred{full}}_n$ (proposition)}

\begin{prop}
\label{prop:fullQ}
For
$
t \in (0, \infty)^N
$
if
$
f^N_{t,t+2^{-N}}
	=
\mathbf{full}^N_{t,t+2^{-N}}
$,
then
for
$
\omega
	\in
B^N_{t}
$
such that
$
\mathbb{Q}^N_t(\{
	\omega
\})
	\ne 0
$,
the following holds.
\[
q_{t+2^{-N}}(\omega 1)
	=
\frac{1}{2}
	-
2^{\frac{N}{2}-1}
\frac{\mu - r}{\sigma} ,
	\; \quad \;
q_{t+2^{-N}}(\omega 0)
	=
\frac{1}{2}
	+
2^{\frac{N}{2}-1}
\frac{\mu - r}{\sigma} .
\]

\end{prop}

\begin{figure}[h]
\center
\[
\xymatrix@C=60 pt@R=5 pt{
&
	\omega 1
\\
\\
	\omega
		\ar @{-} [ruu]
		\ar @{-} [rdd]
\\
\\
&
	\omega 0
\\\\
    B^N_t
        \ar @{}_{ \mathrel{ \rotatebox[origin=c]{90}{$\in$} } } @<+3pt> [dd]
&
    B^N_{t+2^{-N}}
        \ar @{->}_{\mathbf{full}^N_{t,t+2^{-N}}} [l]
        \ar @{}_{ \mathrel{ \rotatebox[origin=c]{90}{$\in$} } } @<+3pt> [dd]
\\\\
	\omega
&
	\omega d_{t+2^{-N}}
        \ar @{|->}_{\mathbf{full}^N_{t,t+2^{-N}}} [l]
}
\]
\caption{
$\mathbf{full}^N_{t+2^{-N},t}$
}
\label{fig:fullN}
\end{figure}

\begin{proof}
By observing
Figure \ref{fig:fullN},
we have
\[
(\mathbf{full}^N_{t,t+2^{-N}})^{-1}(\omega)
	=
\{ \omega 0, \omega 1 \} , 
	\quad
I^N_t(1, \omega) = \{ \omega 1 \} ,
	\quad
I^N_t(0, \omega) = \{ \omega 0 \} .
\]

\noindent
Then, by Theorem \ref{thm:iiffMart},
\[
\mathbb{Q}^N_t(\{\omega\})
	=
c_1 \mathbb{Q}^N_{t+2^{-N}}(I^N_t(1,\omega))
	+
c_0 \mathbb{Q}^N_{t+2^{-N}}(I^N_t(0,\omega))
	=
c_1 \mathbb{Q}^N_{t+2^{-N}}(\{\omega 1\})
	+
c_0 \mathbb{Q}^N_{t+2^{-N}}(\{\omega 0\}) .
\]
Since
\begin{equation*}
\mathbb{Q}^N_{t+2^{-N}}(\{\omega d_{t+2^{-N}}\})
	=
\mathbb{Q}^N_{t}(\{\omega\})
q_{t+2^{-N}}(\omega d_{t+2^{-N}})
\end{equation*}
by
Assumption \ref{asm:asmQ}
and
$
\mathbb{Q}^N_{t}(\{\omega\})
	\ne 0
$,
we have
\begin{equation*}
1 =
c_1 q_{t+2^{-N}}(\omega 1)
	+
c_0 q_{t+2^{-N}}(\omega 0) .
\end{equation*}
Hence,
by Lemma \ref{lem:qbin},
we obtain
\begin{equation*}
q_{t+2^{-N}}(\omega 1)
	=
\frac{1}{2}
	-
2^{\frac{N}{2}-1}
\frac{\mu - r}{\sigma} ,
\quad
q_{t+2^{-N}}(\omega 0)
	=
\frac{1}{2}
	+
2^{\frac{N}{2}-1}
\frac{\mu - r}{\sigma} .
\end{equation*}

\end{proof}

Note that the probability obtained in 
Proposition \ref{prop:fullQ}
does not depend on either
$\omega$
or
$t$ .

\begin{prop}%{[$\mathfrak{f}^{\textred{drop}}^N_t$ (proposition)]}
\label{prop:dropQ}
For
$t \in (0, \infty)^N$,
if
$
f^N_t
	=
\mathbf{drop}^N_{t,t+2^{-N}}
$,
then
for
$
\omega \in B^N_{t-2^{-N}}
$
such that
$
\mathbb{Q}^N_{t-2^{-N}}(\{\omega\})
	\ne 0
$,
the following holds.
\begin{align*}
q_{t}(\omega 1)
	&=
0 ,
\\
q_{t}(\omega 0)
	&=
1 ,
\\
q_{t+2^{-N}}(\omega 01)
	&=
\frac{1}{2}
	-
2^{\frac{N}{2}-1}
\frac{\mu - r}{\sigma} ,
	\\
q_{t+2^{-N}}(\omega 00)
	&=
\frac{1}{2}
	+
2^{\frac{N}{2}-1}
\frac{\mu - r}{\sigma} .
\end{align*}

\end{prop}

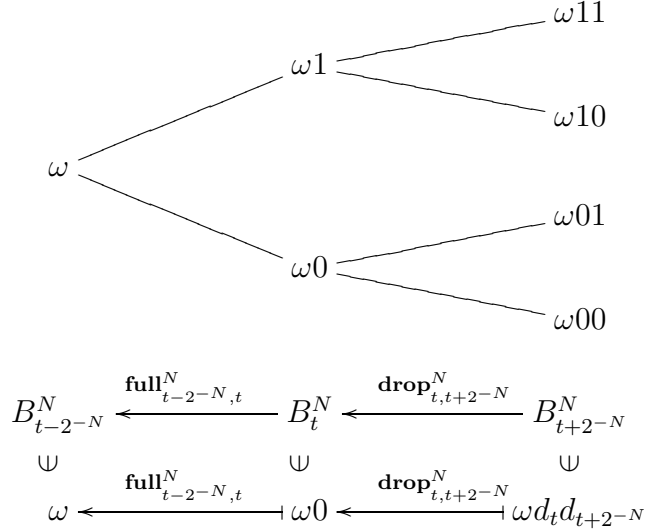
\begin{figure}[h]
\center
\[
\xymatrix@C=60 pt@R=5 pt{
&&
	\omega 11
\\
&
	\omega 1
		\ar @{-} [ru]
		\ar @{-} [rd]
\\
&&
	\omega 10
\\
	\omega
		\ar @{-} [ruu]
		\ar @{-} [rdd]
\\
&&
	\omega 01
\\
&
	\omega 0
		\ar @{-} [ru]
		\ar @{-} [rd]
\\
&&
	\omega 00
\\\\
	B^N_{t-2^{-N}}
        \ar @{}_{ \mathrel{ \rotatebox[origin=c]{90}{$\in$} } } @<+3pt> [dd]
&
    B^N_t
        \ar @{->}_{\mathbf{full}^N_{t-2^{-N},t}} [l]
        \ar @{}_{ \mathrel{ \rotatebox[origin=c]{90}{$\in$} } } @<+3pt> [dd]
&
    B^N_{t+2^{-N}}
        \ar @{->}_{\mathbf{drop}^N_{t,t+2^{-N}}} [l]
        \ar @{}_{ \mathrel{ \rotatebox[origin=c]{90}{$\in$} } } @<+3pt> [dd]
\\\\
	\omega
&
	\omega \textblue{0}
        \ar @{|->}_{\mathbf{full}^N_{t-2^{-N},t}} [l]
&
	\omega d_t d_{t+2^{-N}}
        \ar @{|->}_{\mathbf{drop}^N_{t,t+2^{-N}}} [l]
}
\]
\caption{
$\mathbf{drop}^N_{t,t+2^{-N}}$
followed by
$\mathbf{full}^N_{t+2^{-N},t}$
}
\label{fig:fullNdropN}
\end{figure}

\begin{proof}
By observing
Figure \ref{fig:fullNdropN},
we have
\begin{align*}
(\mathbf{drop}^N_{t,t+2^{-N}})^{-1}(\omega 1) &= \emptyset ,
\\
I^N_t(1, \omega 1) &=
I^N_t(0, \omega 1) = \emptyset ,
\\
(\mathbf{drop}^N_{t,t+2^{-N}})^{-1}(\omega 0) 
	&= 
\{ \omega 00, \omega 01, \omega 10, \omega 11 \} ,
\\
I^N_t(1, \omega 0) &= \{ \omega 01, \omega 11 \} ,
\\
I^N_t(0, \omega 0) &= \{ \omega 00, \omega 10 \} .
\end{align*}
Then, by Theorem \ref{thm:iiffMart},
\[
\mathbb{Q}^N_t(\{\omega 1\})
	=
c_1 \mathbb{Q}^N_{t+2^{-N}}(I^N_t(1,\omega 1))
	+
c_0 \mathbb{Q}^N_{t+2^{-N}}(I^N_t(0,\omega 1))
	=
0 .
\]
Now, since
$
\mathbb{Q}^N_{t}(\{\omega d_{t}\})
	=
\mathbb{Q}^N_{t-2^{-N}}(\{\omega\})
q_{t}(\omega d_{t}) 
$
by Assumption \ref{asm:asmQ},
and
$
\mathbb{Q}^N_{t-2^{-N}}(\{\omega\})
	\ne 0
$,
we have
\begin{equation*}
q_{t}(\omega 1) = 0,
\quad
q_{t}(\omega 0) 
	=
1 - q_{t}(\omega 1) 
	=
1 .
\end{equation*}
Next,
again by
Theorem \ref{thm:iiffMart},
\begin{align*}
\mathbb{Q}^N_t(\{\omega 0\})
	&=
c_1 \mathbb{Q}^N_{t+2^{-N}}(I^N_t(1,\omega 0))
	+
c_0 \mathbb{Q}^N_{t+2^{-N}}(I^N_t(0,\omega 0))
	\\&=
c_1 \big(
	\mathbb{Q}^N_{t+2^{-N}}(\{\omega 01\})
		+
	\mathbb{Q}^N_{t+2^{-N}}(\{\omega 11\})
\big)
	\\&+
c_0 \big(
	\mathbb{Q}^N_{t+2^{-N}}(\{\omega 00\})
		+
	\mathbb{Q}^N_{t+2^{-N}}(\{\omega 10\})
\big) .
\end{align*}
By dividing both sides by
$
	\mathbb{Q}^N_{t-2^{-N}}(\{\omega\})
\ne
	0
$,
we obtain
\begin{align*}
q_t(\omega 0)
	&=
c_1 \big(
	q_t(\omega 0) q_{t+2^{-N}}(\omega 01)
		+
	q_t(\omega 1) q_{t+2^{-N}}(\omega 11)
\big)
	\\&+
c_0 \big(
	q_t(\omega 0) q_{t+2^{-N}}(\omega 00)
		+
	q_t(\omega 1) q_{t+2^{-N}}(\omega 10)
\big) .
\end{align*}
Hence, since
$
q_{t}(\omega 1) = 0
$
and
$
q_{t}(\omega 0) = 1
$,
we get
\begin{equation*}
1
	=
c_1 q_{t+2^{-N}}(\omega 01)
	+
c_0 q_{t+2^{-N}}(\omega 00) .
\end{equation*}
Therefore, by
Lemma \ref{lem:qbin},
\begin{equation*}
q_{t+2^{-N}}(\omega 01)
	=
\frac{1}{2}
	-
2^{\frac{N}{2}-1}
\frac{\mu - r}{\sigma},
	\quad
q_{t+2^{-N}}(\omega 00)
	=
\frac{1}{2}
	+
2^{\frac{N}{2}-1}
\frac{\mu - r}{\sigma} .
\end{equation*}

\end{proof}

\begin{rem}
We have the following remarks for 
%Figure \ref{fig:fn_drop}.
Figure \ref{fig:fullNdropN},
\begin{enumerate}
\item
Since the agent evaluates stock and bond along the function 
$\mathbf{drop}^N_{t,t+2^{-N}}$,
she can recognize only the nodes $\omega 0$, $\omega 01$ and $\omega 00$ 
and can not recognise the nodes $\omega 1$, $\omega 11$ and $\omega 10$. 
We interpret these nodes $\omega 1$, $\omega 11$ and $\omega 10$ as invisible. 

\item
The values
$
q_{t+2^{-N}}(\omega 11)
	\in [0, 1]
$
can be arbitrarily selected,
and
$q_{t+2^{-N}}(\omega 10)$
is computed by
$
1 - q_{t+2^{-N}}(\omega 10)
$.
That is, the probability measure
$\mathbb{Q}^N_{t+2^{-N}}$
is not determined uniquely,
so is not the risk-neutral filtration
$\mathcal{C}^N$.

\item
The probability measure
$\mathbb{Q}^N_{t}$
is not equivalent to the original measure
$\mathbb{P}^N_{t}$.
Therefore, it is not an EMM.

\end{enumerate}

\end{rem}

\begin{prop}
\label{prop:CwllDef}
Both
$\mathbf{full}^N_{t,t+2^{-N}}$
and
$\mathbf{drop}^N_{t,t+2^{-N}}$
are
null-preserving
with respect to
$
\mathbb{Q}^N_{t}
$
and
$
\mathbb{Q}^N_{t+2^{-N}}
$.

\end{prop}

\begin{proof}
Let
$
\omega
	\in
B^N_t
$.
Then, by
Assumption \ref{asm:asmQ},
\begin{equation*}
(
\mathbb{Q}^N_{t+2^{-N}}
	\circ
(\mathbf{full}^N_{t,t+2^{-N}})^{-1}
)
(\omega)
	=
\mathbb{Q}^N_{t+2^{-N}}
(\{
	\omega 1, \omega 0
\})
	=
\mathbb{Q}^N_{t}
(\omega) .
\end{equation*}
Hence,
$\mathbf{full}^N_{t,t+2^{-N}}$
is null-preserving.

Next,
consider the case
when
$
\mathbf{drop}^N_{t,t+2^{-N}}
$.
Then for
$
\omega'
	\in
B^N_{t-2^{-N}}
$,
by Proposition \ref{prop:dropQ},
we have
$
\mathbb{Q}^N_{t}
(\omega' 1)
	= 0
$.
On the other hand,
we get
\begin{equation*}
(
\mathbb{Q}^N_{t+2^{-N}}
	\circ
(\mathbf{drop}^N_{t,t+2^{-N}})^{-1}
)
(\omega' 1)
	=
\mathbb{Q}^N_{t+2^{-N}}
(\emptyset)
	=
0 .
\end{equation*}
Therefore,
$\mathbf{drop}^N_{t,t+2^{-N}}$
is also null-preserving.

\end{proof}

%\begin{thm}[{\cite{ANR_2019_1}}]
\begin{thm}
\label{thm:riskNeuDropFil}
There exists a risk-neutral filtration 
$ \mathcal{C}^N $
for 
the dropped filtration 
$ \mathbf{Drop}_{\alpha, \beta} $.
In this case, 
the probability measure 
$ \mathbb{Q}^N_t $
of the probability space 
$ \mathcal{C}^N(t) $
is not equivalent to
the probability measure
$ \mathbb{P}_{t} $
of 
$ \mathbf{Drop}_{\alpha, \beta}(t) $.
Therefore,
it is not an EMM.
In fact, 
the probability measure
$\mathbb{Q}^N_{t}$
is not uniquely determined. 
Similarly, the risk-neutral filtration
$\mathcal{C}^N$
is not uniquely determined.
\end{thm}
\begin{proof}
Substituting the 
$ q_t $
obtained by 
Propositions \ref{prop:fullQ}
and 
\ref{prop:dropQ}
into 
Assumption \ref{asm:asmQ},
we obtain the probability measure 
$ \mathbb{Q}^N_t $.
On the other hand, 
from Proposition \ref{prop:CwllDef},
the arrows 
$\mathbf{full}^N_{t,t+2^{-N}}$
and 
$\mathbf{drop}^N_{t,t+2^{-N}}$
are null-preserved under 
$ \mathbb{Q}^N_t $.
Therefore, 
we can say that 
$\mathcal{C}^N$
is a filtration. 
Moreover, 
$ \mathbb{Q}^N_t $
clearly satisfies the necessary and sufficient conditions of 
Theorem \ref{thm:iiffMart}
from the way it is constructed. 
Therefore, 
the filtration 
$\mathcal{C}^N$
is a risk-neutral filtration with respect to 
$ \mathbf{Drop}_{\alpha, \beta} $.
By the way, 
in Proposition \ref{prop:dropQ},
$
q_{t+2^{-N}}(\omega 11)
	\in [0, 1]
$
can take any value. 
Then 
$q_{t+2^{-N}}(\omega 10)$
can be computed by 
$
1 - q_{t+2^{-N}}(\omega 11)
$.
That is, 
the probability measure 
$\mathbb{Q}^N_{t+2^{-N}}$
is not uniquely determinable.

\end{proof}

%\vspace{-2.0\baselineskip}

\begin{figure}[h]
\center
\[
%\xymatrix@C=37 pt@R=3 pt{
\xymatrix@C=75 pt@R=3 pt{
&&&
	\textred{\omega 111}
\\
&&
	\textred{\omega 11}
		\ar @{-} [ru]
		\ar @{-} [rd]
\\
&&&
	\textred{\omega 110}
\\
&
	\omega 1
		\ar @{-} [ruu]
		\ar @{-} [rdd]
\\
&&&
	\omega 101
\\
&&
	\omega 10
		\ar @{-} [ru]
		\ar @{-} [rd]
\\
&&&
	\omega 100
\\
	\omega
		\ar @{-} [ruuuu]
		\ar @{-} [rdddd]
\\
&&&
	\textred{\omega 011}
\\
&&
	\textred{\omega 01}
		\ar @{-} [ru]
		\ar @{-} [rd]
\\
&&&
	\textred{\omega 010}
\\
&
	\omega 0
		\ar @{-} [ruu]
		\ar @{-} [rdd]
\\
&&&
	\omega 001
\\
&&
	\omega 00
		\ar @{-} [ru]
		\ar @{-} [rd]
\\
&&&
	\omega 000
\\\\
	B^N_{t-2\cdot2^{-N}}
&
	B^N_{t-2^{-N}}
        \ar @{->}_{\mathbf{full}_{t-2\cdot2^{-N},t-2^{-N}}} [l]
&
    B^N_t
        \ar @{->}_{\mathbf{full}_{t-2^{-N},t}} [l]
&
    B^N_{t+2^{-N}}
        \ar @{->}_{\textblue{\mathbf{drop}^N_{t,t+2^{-N}}}} [l]
}
\]
\caption{Filtration $\mathbf{Drop}_{t-0.5, t+0.5}$}
\label{fig:dropFiltt}
\end{figure}
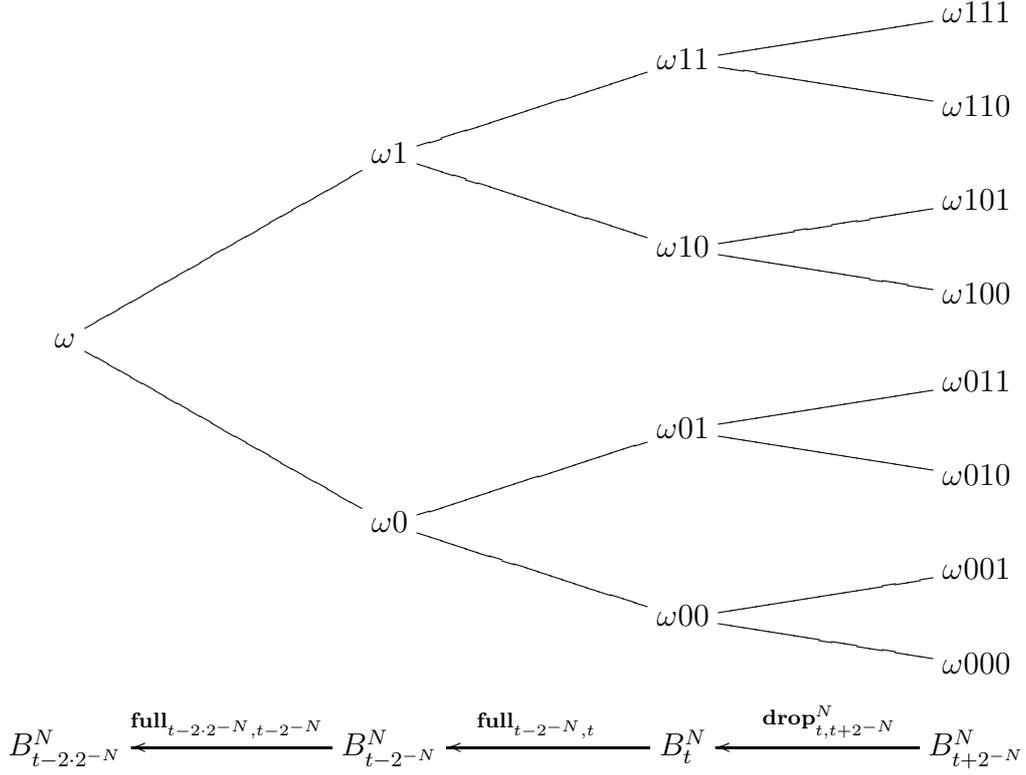

\subsection{Valuation}
\label{sec:valuation}

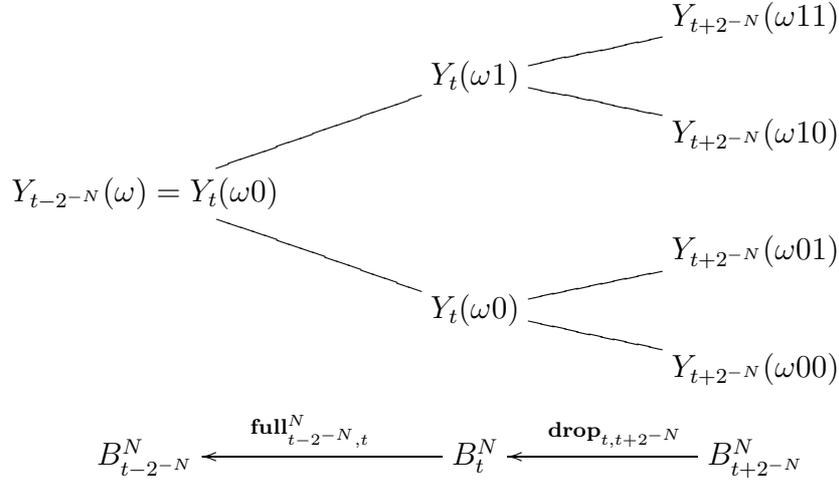
\begin{figure}[h]
\center
\begin{equation*}
%\xymatrix@C=30 pt@R=3 pt{
\xymatrix@C=50 pt@R=3 pt{
&&
	\textred{Y_{t+2^{-N}}(\omega 11)}
\\
&
	\textred{Y_t(\omega 1)}
		\ar @{-} [ru]
		\ar @{-} [rd]
\\
&&
	\textred{Y_{t+2^{-N}}(\omega 10)}
\\
	\textblue{Y_{t-2^{-N}}(\omega) = Y_t(\omega 0)}
		\ar @{-} [ruu]
		\ar @{-} [rdd]
\\
&&
	Y_{t+2^{-N}}(\omega 01)
\\
&
	Y_t(\omega 0)
		\ar @{-} [ru]
		\ar @{-} [rd]
\\
&&
	Y_{t+2^{-N}}(\omega 00)
\\\\
	B^N_{t-2^{-N}}
&
    B^N_t
        \ar @{->}_{\mathbf{full}^N_{t-2^{-N},t}} [l]
&
    B^N_{t+2^{-N}}
        \ar @{->}_{\mathbf{drop}_{t,t+2^{-N}}} [l]
}
\end{equation*}
\caption{Valuation through $\mathbf{Drop}_{t,t}$}
\label{fig:dropVal}
%\label{fig:dropVal}
\end{figure}

Let
$
\textred{\mathcal{C}^N}
 : \mathcal{T}^N \to \Prob
$
be a risk-neutral filtration
and
$
\textred{Y}
 : B^N_{T} \to \mathbb{R}
$
be a payoff
at time
$
\textred{T}
$.
Then,
the price
$
Y_t
$
of
$Y$
at time
$
\textred{t}
$
is
given by the equation
\begin{equation*}
\textred{Y_t}
	:=
E^{\mathcal{C}^N(\iota^N_{t,T})}(
	(b^N_T)^{-1} Y
) 
\end{equation*}
with the unique arrow
$
\textred{\iota^N_{t,T}} : T \to t
$.

That is to say, 
even those who have a dropped subjective filtration
can price Securities 
$Y$.
However, 
additional consideration is needed on how these prices affect the market equilibrium price.

For
$\omega \in B^N_{t-2 \cdot 2^{-N}}$,
you can see in 
Figure \ref{fig:dropVal}
that
at time
$t-2^{-N}$
the value of
$Y_t(\omega 1)$
is discarded and use only the value of
$Y_t(\omega 0)$
for computing 
$Y_{t-2^{-N}}(\omega)$.

%\begin{figure}[h]
%\begin{equation*}
%\xymatrix@C=37 pt@R=0 pt{
%&&
%	\textred{Y_{n+1}(a11)}
%\\
%&
%	\textred{Y_n(a1)}
%		\ar @{-} [ru]
%		\ar @{-} [rd]
%\\
%&&
%	\textred{Y_{n+1}(a10)}
%\\
%	\textblue{Y_{n-1}(a) := Y_n(a0)}
%		\ar @{-} [ruu]
%		\ar @{-} [rdd]
%\\
%&&
%	Y_{n+1}(a01)
%\\
%&
%	Y_n(a0)
%		\ar @{-} [ru]
%		\ar @{-} [rd]
%\\
%&&
%	Y_{n+1}(a00)
%\\\\
%	B_{n-1}
%&
%    B_n
%        \ar @{->}_{f^{full}_{n-1}} [l]
%&
%    B_{n+1}
%        \ar @{->}_{f^{drop}_n} [l]
%}
%\end{equation*}
%\caption{Valuation along $f_n^{drop}$}
%\label{fig:dropVal}
%\end{figure}

%%%%%%%%%%%%%%
%\input{rep.tex} % replication strategy

\subsubsection{Replication Strategies}

Let us investigate the situation where
a given strategy
$
%( \phi_n, \psi_n)_{n = 1, 2, \dots}
(\phi, \psi)
$
becomes a replication strategy of
the payoff 
$Y$
at time
$T$.

\begin{defn}{[Self-Financial Strategies]}
A
\newword{self-financial}
strategy
is 
a strategy
$
(\phi, \psi)
$
satisfying
\begin{equation}
S^N_t \phi_{t+2^{-N}}
	+
b^N_t \psi_{t+2^{-N}}
	=
V_t
\end{equation}
for
every
$
t \in (0, \infty)^N
$.

\end{defn}

For a self-financial strategy
$
(
\phi_t, \psi_t
)_{t \in (0, \infty)^N}
$,
we have:
\begin{align*}
V_{t+2^{-N}}
	=&
S^N_{t+2^{-N}} (\phi_{t+2^{-N}} \circ 
f^N_{t,t+2^{-N}})
	+
b^N_{t+2^{-N}} (\psi_{t+2^{-N}} \circ 
f^N_{t,t+2^{-N}})
	\\=&
(S^N_t \circ 
f^N_{t,t+2^{-N}})
(1 + 2^{-N} \mu + 2^{-\frac{N}{2}}\sigma \xi^N_{t+2^{-N}})
(\phi_{t+2^{-N}} \circ 
f^N_{t,t+2^{-N}})
	\\&
	+
b^N_{t+2^{-N}} 
(
(b^N_t)^{-1}
(V_t - S^N_t \phi_{t+2^{-N}})
\circ 
f^N_{t,t+2^{-N}})
	\\=&
(1 + 2^{-N} \mu + 2^{-\frac{N}{2}}\sigma \xi^N_{t+2^{-N}})
((S^N_t \phi_{t+2^{-N}}) \circ 
f^N_{t,t+2^{-N}})
	+
(1+ 2^{-N} r)
(
(V_t - S^N_t \phi_{t+2^{-N}})
\circ 
f^N_{t,t+2^{-N}})
	\\=&
(2^{-N} \mu - 2^{-N} r + 2^{-\frac{N}{2}}\sigma \xi^N_{t+2^{-N}})
((S^N_t \phi_{t+2^{-N}}) \circ 
f^N_{t,t+2^{-N}})
	\\&
	+
(1+ 2^{-N}r)
(
V^N_t 
\circ 
f^N_{t,t+2^{-N}}) .
\end{align*}
Therefore,
for
$\omega \in B^N_t$
and
$
d_{t+2^{-N}}
	\in
\{
0, 1
\}
$,
\begin{equation}
\label{eq:vnPreEq}
V_{t+2^{-N}}
(\omega d_{t+2^{-N}})
	=
(2^{-N}\mu - 2^{-N}r + 2^{-\frac{N}{2}}\sigma (2 d_{t+2^{-N}} - 1))
S^N_t(\omega_t) 
\phi_{t+V}(\omega_t)
	+
(1+2^{-N} r)
V_t(\omega_t)
\end{equation}
where
\begin{equation}
\omega_t
	:=
f^N_{t,t+2^{-N}}
(\omega d_{t+2^{-N}}) .
\end{equation}

Now let us assume that there exists a function
$
g_t 
	:
B^N_t
	\to
B^N_t
$
such that
$
f^N_{t,t+2^{-N}}
=
	g_t
\circ
	\mathbf{full}_{t,t+2^{-N}}
$ .
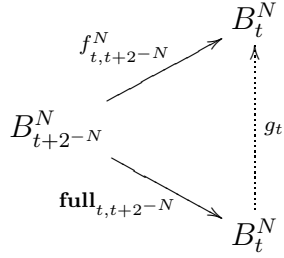
\begin{equation*}
\xymatrix@C=40 pt@R=20 pt{
&
	B^N_t
\\
	B^N_{t+2^{-N}}
		\ar @{->}^{
			f^N_{t,t+2^{-N}}
		} [ru]
		\ar @{->}_{
			\mathbf{full}_{t,t+2^{-N}}
		} [rd]
\\
&
	B^N_t
		\ar @{.>}_{\textred{g_t}} [uu]
}
\end{equation*}
Then
$
f^N_{t,t+2^{-N}}
(\omega d_{t+2^{-N}})
	=
g_t(\omega)
$
for every
$
\omega \in B^N_t
$
and
$
d_{t+2^{-N}}
	\in
\{ 0, 1 \}
$.
So the equation
(\ref{eq:vnPreEq})
becomes
\begin{equation}
\label{eq:vnPreMEq}
V_{t+2^{-N}}
(\omega d_{t+2^{-N}})
	=
(2^{-N}\mu - 2^{-N}r + 2^{-\frac{N}{2}}\sigma (2 d_{t+2^{-N}} - 1))
S^N_t(g_t(\omega)) 
\phi_{t+2^{-N}}(g_t(\omega))
	+
(1+2^{-N}r)
V_t(g_t(\omega)) .
\end{equation}
Hence, we have:
\begin{align}
\label{eq:phiN}
\phi_{t+2^{-N}}(g_t(\omega))
	&=
\frac{
	V_{t+2^{-N}}(\omega 1)
		-
	V_{t+2^{-N}}(\omega 0)
}{
	2^{1-\frac{N}{2}} \sigma S^N_t(g_t(\omega))
}
\\
\label{eq:VN}
V_t(g_t(\omega))
	&=
\frac{
	(2^{\frac{N}{2}}\sigma - \mu + r)
	V_{t+2^{-N}}(\omega 1)
		+
	(2^{\frac{N}{2}}\sigma - \mu + r)
	V_{t+2^{-N}}(\omega 0)
}{
	2^{1 + \frac{N}{2}}\sigma
	(1+ 2^{-N}r)
} .
\end{align}

Therefore,
we can determine the appropriate strategy
$
(\phi_{t+2^{-N}}, \psi_{t+2^{-N}})
$
on 
$
g_t(B^N_t)
	\subset
B^N_t
$
by (\ref{eq:phiN}).
We actually do not care the values of
$
(\phi_{t+2^{-N}}, \psi_{t+2^{-N}})
$
on
$
B^N_t \setminus g_t(B^N_t)
$.

For example,
in the case of
$
f^N_{t,t+2^{-N}}
	=
\mathbf{full}_{t,t+2^{-N}}
$,
the function
$
g_t
	:
B^N_t \to B^N_t
$
satisfies
\begin{equation}
g_t(\omega' d_{t})
	=
\omega' 0
\end{equation}
for all
$
\omega' \in
B^N_{t-2^{-N}}
$
and
$
d_{t}
	\in
\{0, 1\}
$.
Looking at
Figure \ref{fig:dropVal},
values in
the region
$
B^N_t \setminus g_t(B^N_t)
$
are not necessary for computing
$
Y_{t-2^{-N}}(\omega)
$.
Hence,
determining
the values of
$
(\phi_{t+2^{-N}}, \psi_{t+2^{-N}})
$
in
$
g_t(B^N_t)
$
is enough for making the practical valuation.

%%%%%%%%%%%%%%

%\input{e}
%subsecframe{My Subjective Filtration must be Full\textred{!}}
\subsection{Experienced Paths}
\label{sec:expPath}

In this subsection,
we introduce a concept of
experienced paths
that corresponds to a subjective recognition of a person's experience.
%My Experienced Filtration is FULL!

\begin{defn}
\label{defn:expPath}
Let
$
\mathcal{B}^N
	:
\mathcal{T}^N
	\to
\Prob
$
be a filtration
and
$t \in [0, \infty]^N$.

\begin{enumerate}
\item
Define
a function
$
\textred{
e^{\mathcal{B}^N}_t
}
	:
B^N_t
	\to
B^N_t
$
by
\begin{equation*}
\textred{
e^{\mathcal{B}^N}_t
}
(\omega)(s)
	:=
f^N_{s,t}
(\omega)(s)
\end{equation*}
for
$
\omega \in B^N_t
$,
$
s \in (0, t]^N
$
and
$
\textblue{
f^N_{s,t}
}
	:=
\mathcal{B}^N(\iota^N_{s,t})
$.

We call
$
\textblue{
e^{\mathcal{B}^N}_t(\omega)
}
$
an
\newword{experienced path}
of
$\omega$.

%\lspace
\item
$
\textred{
\tilde{\mathcal{B}}^N_t
}
	:= \{
e^{\mathcal{B}^N}_t(\omega)
	\mid
\omega \in B^N_t
\}
$.

%\lspace
\item
$
\textred{
\tilde{\mathcal{F}}^N_t
}
	:=
2^{
	\tilde{\mathcal{B}}^N_t
}
$.

%\lspace
\item
$
\textred{
\tilde{\mathbb{P}}^N_t
}
	:=
\mathbb{P}^N_t
	\circ
(e^{\mathcal{B}^N}_t)^{-1}
$.

\[
\xymatrix@C=30 pt@R=7 pt{
&
	B^N_t
		\ar @{->}^{
			e^{\mathcal{B}^N}_t
		} [r]
&
	\tilde{B}^N_t
\\
	[0,1]
&
	\mathcal{F}^N_t
		\ar @{->}_{
			\mathbb{P}^N_t
		} [l]
&
	\tilde{\mathcal{F}}^N_t
		\ar @{->}_{
			(e^{\mathcal{B}^N}_t)^{-1}
		} [l]
}
\]

\item
$
\textred{
\bar{\tilde{B}}^N_t
}
	:=
(
	\tilde{B}^N_t,
	\tilde{\mathcal{F}}^N_t,
	\tilde{\mathbb{P}}^N_t
)
$.

\item
For
$
s, t
	\in
[0, \infty)^N
	\;
(s \le t)
$,
$
\textred{
	\tilde{f}^N_{s,t}
}
	:
\bar{\tilde{B}}^N_t
	\to
\bar{\tilde{B}}^N_s
$
is a function defined by
\begin{equation*}
\textred{
	\tilde{f}^N_{s,t}
}
	:=
\textblue{\mathbf{full}}^N_{s,t}
	\mid_{
		\tilde{B}^N_t
	} .
\end{equation*}
\end{enumerate}

\end{defn}

\begin{prop}%{[My Experienced Filtration is FULL!]}
A correspondence
$
\textred{
	\tilde{\mathcal{B}}^N
}
	:
\mathcal{T}^N
	\to
\Prob
$
defined by

\begin{equation*}
	\tilde{\mathcal{B}}^N
(t)
	:=
\bar{\tilde{B}}^N_t
	\quad \textrm{and} \quad
\textred{
	\tilde{\mathcal{B}}^N
}(
	\iota^N_{s,t}
)
	:=
\tilde{f}^N_{s,t}
\end{equation*}
is a functor, that is, 
a $\mathcal{T}^N$-filtration.

\end{prop}

\begin{exmp}{[Experienced Paths for
$ \mathcal{B}^N := \mathbf{Drop}^N_{\textblue{\frac{3}{8}}, \textblue{\frac{5}{8}}} $
]}
Let
$ \mathcal{B}^N := \mathbf{Drop}^N_{\textblue{\frac{3}{8}}, \textblue{\frac{5}{8}}} $
and
$d_i \in \{ 0, 1 \}$
for
$i \in \mathbb{N}$.
Then, 
as seen in 
Figure \ref{fig:expPath},
we have
\begin{align*}
e^{\mathcal{B}^2}_1
(d_1 d_2 d_3 d_4)
	&=
d_1 0 d_3 d_4 ,
\\
e^{\mathcal{B}^3}_1
(d_1 d_2 d_3 d_4 d_5 d_6 d_7 d_8)
	&=
d_1 d_2 0 0 0 d_6 d_7 d_8 .
\end{align*}
\begin{figure*}[t]
\center
\[
%\xymatrix@C=20 pt@R=7 pt{
\xymatrix@C=20 pt@R=10 pt{
	{*}
&
	B^2_0
		\ar @{}^-{\in} @<-6pt> [l]
&
	B^3_0
		\ar @{}_{ \mathrel{ \rotatebox[origin=c]{180}{$\in$} } } @<+6pt> [r]
&
	{*}
&
	0
\\
&
&
	B^3_{1/8}
		\ar @{->}_{\mathbf{full}} [u]
		\ar @{}_{ \mathrel{ \rotatebox[origin=c]{180}{$\in$} } } @<+6pt> [r]
&
	\textred{d_1}
&
	1/8
		\ar @{->} [u]
\\
	\textred{d_1}
&
	B^2_{1/4}
		\ar @{->}_{\mathbf{full}} [uu]
		\ar @{}^-{\in} @<-6pt> [l]
&
	B^3_{2/8}
		\ar @{->}_{\mathbf{full}} [u]
		\ar @{}_{ \mathrel{ \rotatebox[origin=c]{180}{$\in$} } } @<+6pt> [r]
&
	d_1 \textred{d_2}
&
	2/8
		\ar @{->} [u]
\\
&
&
	B^3_{3/8}
		\ar @{->}_{\mathbf{full}} [u]
		\ar @{}_{ \mathrel{ \rotatebox[origin=c]{180}{$\in$} } } @<+6pt> [r]
&
	d_1 d_2 \textred{0}
&
	3/8=\textblue{\alpha}
		\ar @{->} [u]
\\
	d_1 \textred{0}
&
	B^2_{2/4}
		\ar @{->}_{\mathbf{full}} [uu]
		\ar @{}^-{\in} @<-6pt> [l]
&
	B^3_{4/8}
		\ar @{->}_{\textblue{\mathbf{drop}}} [u]
		\ar @{}_{ \mathrel{ \rotatebox[origin=c]{180}{$\in$} } } @<+6pt> [r]
&
	d_1 d_2 d_3 \textred{0}
&
	4/8
		\ar @{->} [u]
\\
&
&
	B^3_{5/8}
		\ar @{->}_{\textblue{\mathbf{drop}}} [u]
		\ar @{}_{ \mathrel{ \rotatebox[origin=c]{180}{$\in$} } } @<+6pt> [r]
&
	d_1 d_2 d_3 d_4 \textred{0}
&
	5/8=\textblue{\beta}
		\ar @{->} [u]
\\
	d_1 d_2 \textred{d_3}
&
	B^2_{3/4}
		\ar @{->}_{\textblue{\mathbf{drop}}} [uu]
		\ar @{}^-{\in} @<-6pt> [l]
&
	B^3_{6/8}
		\ar @{->}_{\textblue{\mathbf{drop}}} [u]
		\ar @{}_{ \mathrel{ \rotatebox[origin=c]{180}{$\in$} } } @<+6pt> [r]
&
	d_1 d_2 d_3 d_4 d_5 \textred{d_6}
&
	6/8
		\ar @{->} [u]
\\
&
&
	B^3_{7/8}
		\ar @{->}_{\mathbf{full}} [u]
		\ar @{}_-{ \mathrel{ \rotatebox[origin=c]{180}{$\in$} } } @<+6pt> [r]
&
	d_1 d_2 d_3 d_4 d_5 d_6 \textred{d_7}
&
	7/8
		\ar @{->} [u]
\\
	d_1 d_2 d_3 \textred{d_4}
&
	B^2_{1}
		\ar @{->}_{\mathbf{full}} [uu]
		\ar @{}^-{\in} @<-6pt> [l]
&
	B^3_{1}
		\ar @{->}_{\mathbf{full}} [u]
		\ar @{}_-{ \mathrel{ \rotatebox[origin=c]{180}{$\in$} } } @<+6pt> [r]
&
	d_1 d_2 d_3 d_4 d_5 d_6 d_7 \textred{d_8}
&
	1=\textblue{t}
		\ar @{->} [u]
}
\]
\caption{Experienced Paths for
$ \mathcal{B}^N := \mathbf{Drop}^N_{\textblue{\frac{3}{8}}, \textblue{\frac{5}{8}}} 
$}
\label{fig:expPath}
\end{figure*}
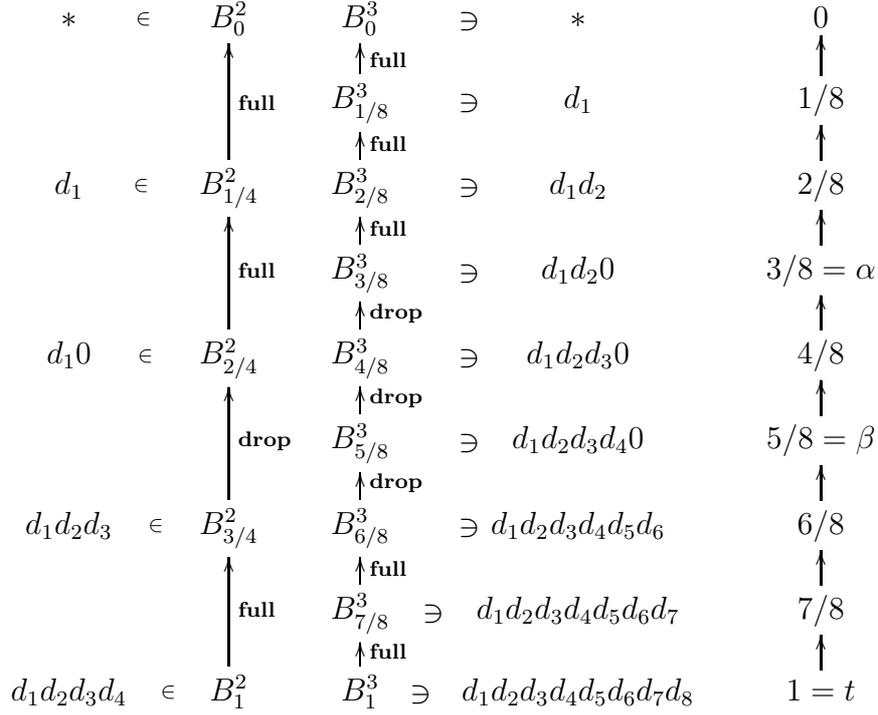

\end{exmp}

\begin{thm}
\label{thm:expPathNat}
The correspondence
$
\textred{
	e^{\mathcal{B}^N}
}
	:
\mathcal{B}^N
	\to
\tilde{\mathcal{B}}^N
$
is a \textblue{natural transformation}.
That is,
for
$
s, t
	\in
[0, \infty)^N
	\;
(s \le t)
$,
the following diagram commutes:
\[
\xymatrix@C=40 pt@R=10 pt{
	\mathcal{T}^N
&
	\mathcal{B}^N
		\ar @{->}^-{
			\textred{e^{\mathcal{B}^N}}
		} [r]
&
	\textred{
		\tilde{\mathcal{B}}^N
	}
\\
	s
&
	\bar{B}^N_s
		\ar @{->}^-{
			e^{\mathcal{B}^N}_s
		} [r]
&
	\bar{\tilde{B}}^N_s
\\\\
	t
		\ar @{->}^{\iota^N_{s,t}} [uu]
&
	\bar{B}^N_t
		\ar @{->}_-{
			e^{\mathcal{B}^N}_t
		} [r]
		\ar @{->}^{f^N_{s,t}} [uu]
&
	\bar{\tilde{B}}^N_t
		\ar @{->}_{
			\textred{
			\tilde{f}^N_{s,t}
			}
		} [uu]
}
\]

\end{thm}

\begin{proof}
For
$
\omega \in B^N_t
$
and
$
u \in (0, s]^N
$,
\begin{align*}
\tilde{f}^N_{s,t}
(
e^{\mathcal{B}^N}_t (\omega)
)
(u)
	&=
(
	\mathbf{full}^N_{s,t}
		\mid_{\tilde{B}^N_t}
)
(
e^{\mathcal{B}^N}_t (\omega)
)
(u)
	\\&=
	\mathbf{full}^N_{s,t}
(
e^{\mathcal{B}^N}_t (\omega)
)
(u)
	\\&=
(
e^{\mathcal{B}^N}_t (\omega)
	\mid_{(0,s]^N}
)
(u)
	\\&=
e^{\mathcal{B}^N}_t (\omega)
(u)
	=
f^N_{u,t}(\omega)(u) .
\end{align*}
On the other hand,
\begin{equation*}
e^{\mathcal{B}^N}_s
(
	f^N_{s,t}(\omega)
)
(u)
	=
f^N_{u,s}(
	f^N_{s,t}(\omega)
)
(u)
	=
f^N_{u,t}(\omega)(u) .
\end{equation*}
\end{proof}

%\frametitle{Implication of the Theorem}
Here is an implication of 
Theorem \ref{thm:expPathNat}: 
The person who dropped her memory
believes that her memory is 
\textblue{perfect}
(\textblue{full}),
while others observe that
she lost her memory.

Lastly, we mention the fact that
in a case the given filtration is full, 
experienced paths coincide with objective paths.

\begin{prop}
\label{prop:fullExpIsObjective}
If 
$
\mathcal{B}^N
	=
\mathbf{Full}^N
$,
then
$
\tilde{\mathcal{B}}^N
	=
\mathcal{B}^N
$.
\end{prop}

\section{Concluding Remarks}

%We formulated an infinitely growing sequence of binomial probability spaces
%in the category $\Prob$.
%We gave some concrete (possibly distorted) filtrations.
%We determined the shape of the risk-neutral filtrations to the above examples.
%We showed the valuations of claims given at time $T$
%through the distorted filtrations,
%and provided a replication strategy implementing the valuation.

%\input{c}
In this paper, 
we proposed the concept of generalized filtration. 
It is an extended filtration that goes beyond the conventional framework of monotonically increasing information sequences 
and allows the development of information to not only increase, 
but also to decrease or be twisted.
It is an extended concept, 
just like the subjective probability measure attributed to an individual, 
of a subjective filtration as a history of personal information evolution. 
%It can represent filtration. 
A natural interest is to see 
how far conventional theories of stochastic analysis and control can be developed under such generalized filtration.

In this paper, 
as an example of an application, 
in addition to conventional filtration (classical filtration) in a binomial asset price model, 
we introduce a dropped filtration with loss of memory for a certain period of time 
to see whether individuals with the latter as her subjective filtration can in any sense price securities. 
%Discussed. 
This resulted in the question of whether there is a risk-neutral filtration corresponding to this subjective filtration. 
We have shown the existence of such a filtration. 
However, the obtained risk-neutral filtration is not uniquely determined, 
unlike the classical risk-neutral probability measure observed in a complete market. 
This means that a market with such a generalized filtration is not complete 
(at least for individuals who have such a filtration as a subjective filtration). 
For other subjective filtrations not discussed in this paper, 
it is possible that there may be no risk-neutral probability measure. 
How equilibrium market prices are determined in such cases may be 
one of
important themes for future research.

Needless to say, 
the application of generalized filtrations shown in this paper is only one example, 
and many other applications are possible. 
As mentioned above, 
generalized filtrations can be used to 
develop conventional theories of stochastic control and stochastic differential equations. 
For example,
it can be used to transform a problem that is not time-consistent under classical filtration 
into a time-consistent problem
by twisting the filtration.
The theory of filtration enlargement used for credit risk calculation and insider trading analysis in finance 
may be able to be considered in the framework of generalized filtration \cite{AJ2017}.
Furthermore, 
in order to study the relationship between a filtration and related risk-neutral filtration, 
or filtrations defined on several different time domains, 
it is necessary to consider the transformation and convergence of filtrations in a space of filtrations.

%%%%%%%%%%%%%%% End of Main Text %%%%%%%%%%%%%%%%%%

%=================================================================
% References
%=================================================================

\nocite{maclane1997}
\nocite{CK2012DM}

\def \EmbedBib  {1}

\if \EmbedBib  0
\bibliographystyle{apalike}
\bibliography{../../taka_e}
\else % \EmbedBib  0

\fi % \EmbedBib  0

%=================================================================

\end{document}